\documentclass[a4paper,12pt]{article}
\usepackage{biblatex}
\usepackage{graphicx} 
\usepackage[margin=1in]{geometry}
\usepackage{hyperref}
\usepackage{xurl}
\hypersetup{breaklinks=true}
\usepackage{booktabs}
\usepackage{mathtools}
\usepackage{amsmath}
\usepackage{amsthm}
\usepackage{amssymb}
\usepackage{amsfonts}
\usepackage{amssymb}
\usepackage{bbm}
\usepackage{calrsfs}
\usepackage{graphicx}
\usepackage{titling}
\usepackage{multicol}
\usepackage[font=small]{caption}
\usepackage{sgame}
\usepackage{longtable}
\usepackage{chngcntr}
\usepackage[capposition=top]{floatrow}
\newcommand{\citelink}[2]{\hyperlink{cite.\therefsection @#1}{#2}}
\newcommand{\qcite}[1]{\citelink{#1}{\citeauthor{#1} (\citeyear{#1})}}
\newcommand{\qcitenp}[1]{\citelink{#1}{\citeauthor{#1} \citeyear{#1}}}

\hypersetup{
    colorlinks=true,
    linkcolor=blue,
    filecolor=magenta,      
    urlcolor=cyan
}

\usepackage{enumitem} 
\usepackage{calc}

\usepackage[framemethod=tikz,%
innerrightmargin=\parindent,%
innerleftmargin=\parindent,%
skipabove=0.6\baselineskip,%
skipbelow=0.6\baselineskip,%
innertopmargin=\parindent,%
innerbottommargin=\parindent]{mdframed}
\newlength{\defparindent}
\newmdenv[linewidth=0.4pt,%
linecolor=black,%
backgroundcolor=white]{singleframed}

\newenvironment*{singleframedindent}{\begin{singleframed}\setlength{\parindent}{\defparindent}\ignorespaces}{\end{singleframed}}

\usepackage{tikz}
\usetikzlibrary{arrows}
\usetikzlibrary{positioning}

\newtheorem{theorem}{Theorem}
\newtheorem{proposition}{Proposition}
\theoremstyle{definition}

\newtheorem{lemma}{Lemma}
\newtheorem{corollary}{Corollary}

\newtheorem{example}{Example}

\title{Matching Design with Algorithms and Applications to Foster Care}
\author{Terence Highsmith II$^1$}
\date{Version: \today \\ Preliminary Draft. PLEASE DO NOT CITE.}

\begin{document}

\maketitle

\normalsize
 \begin{center}
\vspace{20pt}
\textbf{Abstract}\\ 
\begin{flushleft}
We study the problem of an organization that matches agents to objects where agents have preference rankings over objects and the organization uses algorithms to construct a ranking over objects on behalf of each agent. Our new framework carries the interpretation that the organization and its agents may be misaligned in pursuing some underlying matching goal. We design matching mechanisms that integrate agent decision-making and the algorithm by prioritizing matches that are unanimously agreeable between the two parties. Our mechanisms also satisfy restricted efficiency properties. Subsequently, we prove that no unanimous mechanism is strategy-proof but that ours can be non-obviously manipulable. We generalize our framework to allow for any preference aggregation rules and extend the famed Gibbard-Satterthwaite Theorem to our setting. We apply our framework to place foster children in foster homes to maximize welfare. Using a machine learning model that predicts child welfare in placements and a novel "elicited preferences" experiment that extracts real caseworkers' preferences, we empirically demonstrate that there are important match-specific welfare gains that our mechanisms extract that are not realized under the status quo. \\
\textit{Journal of Economic Literature Classification:} C78, D47, D71
\end{flushleft}

\vspace{20pt}

\small  
$^1$) Reach at\\
terence.highsmith@kellogg.northwestern.edu
\end{center}

\newpage

\normalsize

\section*{1 \hspace{5pt} Introduction}

The traditional goal of matching theory is to use agents' preferences to achieve an assignment of objects to agents or agents to agents that is efficient, fair, and sustainable, or some combination of all three. However, in some organizational contexts, matching based on individual decision-maker preferences can hamper an institution's goals. Decision-makers, even when attempting to pursue the organizational objective, can fall prey to informational asymmetries, cognitive limitations, or a lack of experience. Simultaneously, even if the organization has access to an algorithm that can outperform the decision-makers, it may be infeasible to implement in practice due to legal, ethical, and trust concerns. 

In this paper, we develop one-sided matching mechanisms that input \textit{two} preference rankings. We take it as given that there is some underlying welfare attained when any decision-maker matches to any object. We represent decision-maker (DM) preferences over objects, which might or might not align with welfare, as one set of preference rankings. The second set of preference rankings---which we call evaluations---are the ordinal representation of an algorithm's predictions of welfare between DM-object matches. We design mechanisms that satisfy a fairness notion that we call \textit{unanimity}, two weaker efficiency properties for unanimous matching mechanisms, and non-obvious manipulability under some conditions. Following this, we apply our mechanisms to the child welfare domain. 

In the United States, the foster care system is responsible for caring for children found to be victims of abuse or neglect at the hands of their parents. In 2022, there were 368,530 children in the foster care system (\qcitenp{afcars}). Economists have evaluated numerous policy interventions to improve outcomes in foster care (\qcitenp{bald-2022}). However, one critical aspect of foster care has gone unnoticed: how should a child be matched to a temporary foster home? 

Existing processes for matching children to foster home are highly varied across U.S. counties, burdened with inefficiencies, and leave practitioners longing for better systems (\qcitenp{zeijlmans-2016}). Therefore, we cannot offer a uniform characterization of a child's journey from her home of origin to a foster home. Generally speaking, children that are suspected to be victims of abuse or neglect are often reported to the local child welfare agency by community members. The child welfare agency sends an investigator to the child's home, and, if the investigator substantiates the allegations, then the child is removed from her home. The agency assigns a caseworker (DM) to represent the child, and the child can be kept in custody of the county for up to several days before she is matched to a temporary foster home (object). The caseworker can use various means to search for a suitable match for the child, and this step in the process is the most heterogeneous step across counties. In most counties, caseworkers sequentially decide to match the child that the caseworker represents to the caseworker's most preferred available home, where the order is simply the first (or next) caseworker to initiate matching. We describe this time-ordered serial dictatorship as the "status quo".

Caseworkers representing children have preferences over the homes that a child can be matched to. Often, these preferences can legitimately be in the best interests of the child and encode important, qualitative information that cannot be easily subsumed into algorithms that predict a child's persistence in a home. On the contrary, some caseworkers might have erroneous judgment causing them to systemically choose homes where children will not persist. Nevertheless, the human element in the matching process is essential, and most counties are unwilling---justifiably so---to fully replace caseworkers with algorithms that decide placements. One objective goal in the matching process is \textit{placement stability} (\qcitenp{carnochan-2013}). A child's placement in a foster home is stable if the child remains in the home until the child is reunified with their home of origin, the foster home permanently adopts the child, or the child ages out of the foster care system. We refer to this as \textit{persistence} to avoid confusion with the matching theory concept of stability. Persistence is not common\footnote{The average number of placements for a child was 2.56 in 2015; see \qcite{cortes-2021}.}, and there are widespread efforts to improve persistence for children as more placements are strongly associated with trauma and worse outcomes for youth in foster care. 

We model this problem as a one-sided matching market where children are strategic agents allocated to homes (objects)\footnote{In practice, homes can have dichotomous preferences. However, this two-sided matching problem reduces to a one-sided matching problem where a child finds a home acceptable only if the child and home find the match acceptable.}. Children have strict ordinal preferences $\succ_c$ over homes\footnote{As aforementioned, caseworkers are the DMs and actually hold preferences, but we refer to them as children in exposition for simplicity. One can think of the caseworkers as analogous to parents in school choice.}. The matchmaker has strict ordinal preferences $>_c$ over homes that a given child can be placed in. We call these evaluations. We introduce a fairness concept termed \textit{unanimity}. A home $h$ is unanimously acceptable for $c$ if $h \succ_c \varnothing$ and $h >_c \varnothing$. A child $c$ has a unanimous improvement over $h$ if there exists another $h'$ where $h' \succ_c h$ and $h' >_c h$. A home is \textit{unanimous} for a child if the child has no unanimous improvement over it. A matching $\mu$ is unanimous if there is no other matching which can weakly expand the set of children matched to unanimous homes and the set of children matched to unanimously acceptable homes. 

Unanimity is a form of preference aggregation, and our particular method of aggregation prioritizes allocations where the matchmaker reduces matches where a DM and the algorithm would jointly agree that another match is superior. This desideratum is fair under the prior that these strong agreements between humans and algorithms should be respected over simply improving preferences or evaluations alone. Simultaneously, unanimity provides robustness in cases where one decision model--the human or the algorithm---is fundamentally flawed but the other is more accurate in predicting welfare. Unanimity is useful when institutional constraints prohibit fully automating the matching process and require respecting human preferences as is the case in foster care. We show a simple equivalence between unanimity and an efficient matching with appropriately aggregated preferences (Proposition \ref{proposition:unanimous-reduction}). 

The equivalence shows that the Serial Dictatorship (SD) mechanism appropriately modified for indifference (Serial Dictatorship with Indifference or SDI) generates unanimous matchings (Proposition \ref{proposition:sdi-unanimous}). Many matchings can be unanimous, and a unanimous and efficient matching can fail to exist. Thus, we impose further efficiency restrictions. We tweak the Top Trading Cycles mechanism to find Pareto improvements on any unanimous matching while maintaining unanimity. Still, unanimous matchings only \textit{avoid} unanimously disagreeable matchings. Our second fairness criteria---unimprovable matchings---guarantees that there is no child that can unanimously improve unless some other child would unanimously object. Our third mechanism, Adaptive SDI, is a simple yet novel modification of SDI that constructs unimprovable matchings.

Turning to strategic incentives, we show that no unanimous (unimprovable) mechanism can be strategy proof (Proposition \ref{proposition:unanimous-sp}). Worse, we show that every unanimous (unimprovable) mechanism is obviously manipulable (Theorem \ref{theorem:obvious-manipulability}). The manipulations are truncation strategies. Nevertheless, we restore some positive results regarding the strategic incentives of (A)SDI. When the number of homes is smaller than the number of children, as is typically the case in foster care, (A)SDI is non-obviously manipulable.

Our results suggest that we might be able to attain a limited version of strategy-proofness that we term group robustness. Group robustness requires that there is no profitable manipulation that can improve all children with respect to either $\succ$ or $>$. This concept is important because the matchings we select may not be efficient with respect to either order which can afford opportunities for children or the matchmaker to improve without harming anyone. We generalize our model to examine whether or not \textit{any} mechanism and aggregated preferences can satisfy group robustness. Unfortunately, we find a negative result: any strategy-proof and efficient mechanism satisfying mild regularity conditions that uses aggregated preferences cannot be group robust unless the aggregation is equivalent to $\succ$ or $>$ (Theorem \ref{theorem:gs-aggregation}).

The following section is still under construction. We design an experiment to understand how the structure of caseworker preferences interacts with unanimity. Our experiment invites real caseworkers to report their preferences to a novel device for the assignment of children to homes in a simulated market where, based on a nationally representative panel of foster care placements, we build a machine learning model to predict persistence in the simulated market. Our experimental device is specifically designed to elicit preferences using simple instruments like surveys; it is strategy proof (Proposition \ref{proposition:rsa-sp}) and features truth-telling as the unique, non-trivial weakly dominant strategy (Proposition \ref{proposition:rsa-unique}). Caseworkers receive monetary payoffs when the children that they represent persist. Our choice experiment elicits caseworkers' true preferences similar to the \qcite{budish-2021} elicited preferences framework, and we increment on their approach by incentivizing elicitation. We input these preferences into our mechanisms to estimate the treatment effect on average persistence in the simulated markets. In a preliminary simulation where caseworkers racially match children and foster homes, the placements under SD using caseworker preferences are barely better than random assignment. SDI increases average persistence by 7.8 pp, empirically demonstrating that unanimity alone improves welfare when the algorithm is more accurate than humans. UTTC improves over SDI by a marginal 0.6 pp. When the out-of-sample error is fitted to the error in the test sample, average persistence still increases by 7 pp (15\%). Caseworkers' preferences do not align with predicted persistence, and these results suggest that they do not align with \textit{actual} persistence. Moreover, we are able to demonstrate that there are match-specific benefits to persistence, a fact not observed in previous work.

\textit{Literature Contribution}---As far as we are aware, no other literature on one-sided matching has ever considered the problem where two rankings may be involved for each agent. Theoretically, \qcite{afacan-2022}'s model is a case of our general matching with preference aggregation in disguise. Their work considers an assignment of arbiters to pairs of agents. They aim for "depth optimal" assignments that maximize the minimum rank of an arbiter across both agents' preference rankings. One can reformulate this problem by aggregating the preferences of each pair using the Max Rank rule (Example \ref{example:regular-rule}). Applying our framework immediately yields their depth-optimal assignments, and Theorem \ref{theorem:gs-aggregation} predicts their negative strategy-proofness results. 

In the adjacent literature on social choice theory, many such issues involving preference aggregation arise. We see our work as a near direct extension of \qcite{danan-2016} to matching theory. Their unambiguous Pareto dominance property emphasizes choosing acts where agents unambiguously agree in the face of ex-ante uncertainty about what social outcomes the acts entail. Conceptually, unanimity is quite similar to this property. The main difference lies in the fact that unanimity is appropriate for application to matching allocations rather than a society's preferences. We remain agnostic as to the formation of the rankings whereas the aforementioned work uses assumptions on preferences to derive important main results. Furthermore, we integrate preference aggregation with matching theory, showing that unanimity is equivalent to efficiency with respect to a preference aggregation. Our work then extends \qcite{gibbard-1973} to our setting showing that there is little hope of preempting strategic incentives.

Additionally, this paper intersects results on Pareto-constrained welfare maximization in the computer science literature. Our evaluation ranking can be reformulated as a cardinal match valuation. In this context, \qcite{biro-2020} study the computational complexity of finding a Pareto efficient matching of highest welfare, where welfare in their case is a sum of match valuations and, in a large part of their work, where the match valuation depends only on one side. They find that this problem is NP-hard in general and prove stringent conditions under which it is computationally feasible. Following this, they define a general linear programming solution that may require permuting all matchings to identify the constrained welfare maximizing matching. They also analyze some conditions where positive strategy proofness results can be proven. We contribute through taking a fundamentally different stance from welfare maximization, i.e., a unanimous matching need not be efficient with respect to the evaluations. Our work still touches on theirs as the equivalence between unanimity and unanimous-preference efficiency implies that finding a unanimous and efficient matching is a case of Pareto-constrained welfare maximization. We use separate efficiency notions for unanimous matchings that do not require a solution to the aforementioned problem. Other works by \qcite{cechlarova-2017} and \qcite{saban-2015} explore the computational tractability of problems inherent in computing efficient and unanimous matchings and motivate some of our restrictions on the problem.

A few additional papers intersecting matching theory and child welfare relate to our work. In a similar matching-theoretic approach to foster care, \qcite{highsmith-2024} develops matching mechanisms for expedient adoptions from foster care where dynamic incentives play a far greater role than in matching children to temporary foster homes. \qcite{slaugh-2016} analyze the matching recommendation system that had previously been abandoned by the Pennsylvania child welfare system and take an operations approach to reform it. They show tangible benefits in increasing the number of placements. \qcite{cortes-2021} estimates caseworker preferences and counterfactual persistence for children under various different policies using an empirical matching framework. Some policies that focus on increasing market thickness, like matching across larger geographic regions, can significantly boost persistence. Notably, \qcite{baron-2024} study the problem of reassigning child investigations to investigators while maintaining a sort of status quo. They coin the term "mechanism reform"---changing the status quo mechanism while maintaining at least the same level of welfare for all agents---which relates to our problem in that unanimous matchings incorporate the existing preferences of DMs.  

\textit{Centralization in Practice}---Foster care in the United States is an interesting application of matching theory because, at present, there are no counties that we are aware of that utilize \textit{any} formal matching mechanism to assign children to homes. One simple reason why this could be the case is that the status-quo, as we have described it, is equivalent to a Pareto efficient matching mechanism. Therefore, all mechanisms developed in matching theory, save those that we described in our literature review, would have nearly or exactly the same efficiency as the status-quo. Furthermore, two-sided matching mechanisms are not well-fit for foster care. This is because the set of children to be matched is fluid and constantly changing; the preferences that homes report, then, must similarly constantly update. One could construct preferences based on child observable characteristics which homes report once-and-for-all, but this would require additional theoretical developments to reach the same desirable properties as in the standard case. In working on this and related projects, we have conducted interviews with practitioners to understand the matching process in foster care. Our interviews indicate that individual caseworkers face a relatively easy problem: choose the best perceived home for a child. However, when practitioners aggregate this decision-making rule to the entire market, they find that persistence falls short of their goals. We show in this work that implementing a unanimous mechanism would be distinct from a Pareto efficient mechanism and could significantly ameliorate welfare losses.

\textit{Paper Organization}---In Section 2, we detail our model preliminaries and mechanism properties. In Section 3, we give theoretical results on unanimous matchings and mechanisms, their non-strategic properties, and their strategic properties. We also present existence proofs for other mechanisms satisfying our properties. Finally, we construct a general model we term matching with preference aggregation that shows an impossibility result regarding incentive-compatibility even if we modify the concept of unanimity. In Section 4, we detail our experiment design and results. Last, we conclude in Section 5 with practical considerations.

\section*{2 \hspace{5pt} Preliminaries}

We model a one-sided matching market. Let $A$ be a set with an arbitrary, finite number of agents. $C$ is the set of all children and $H$ is the set of all prospective homes for children, thus $A = C \cup H$. The matchmaker's task is to assign each generic child $c \in C$ to a generic home $h \in H$. Formally, the matching market is an undirected, bipartite graph with labels $M = (C,H,E)$ where $E$ is the edge set representing feasible assignments. Let $\gamma : E \rightarrow \{ \{ c,h \} : c \in C, h \in H \}$ represent the incidence function for $M$ and $E_h(\mu) = \{ d \in \mu : h \in \gamma(d) \}$ where $\mu \subseteq E$. A \textit{feasible matching or assignment} is some $\mu \subseteq E$ that is one-to-one\footnote{We do not pursue a capacitated model. One might imagine that some homes are willing to take in multiple foster children, which can be true, but this rarely happens in a single matching market. Homes will typically only be offered one placement at any given time. We may capture the separate issue that some children require greater than unit capacity (sibling groups) by incorporating this into preferences, where a home is only acceptable if it can accommodate the group.}.

We say that child $c$ is assigned to or placed with home $h$ under $\mu$ if $\{ c, h \} \in \mu$. For convenience, we use the shorthand notation with $\mu$ as a function where $\mu : A \rightarrow A \cup \varnothing$ such that $\mu(c) = h \iff c \in \mu(h) \iff \{ c, h \} \in \mu$. We say that $c$ is unassigned or unplaced if $c \notin \gamma(d)$ for any $d \in \mu$ and write $\mu(c) = \varnothing$. We will also write that $\mu(\varnothing) = \varnothing$.

$P$ is the set of preference orderings (total preorders) over $H \cup \varnothing$. Each child $c$ has a ranking $\succ_c \hspace{5pt} \in P$ over prospective homes. There are no externalities. Each child only cares about her own assignments. We write $h \succ_c h'$ if and only if $c$ ranks $h$ higher than $h'$. If $\varnothing \succ_c h$, we say that $h$ is \textit{unacceptable} to $c$. We also assume that children's preferences are strict and write that $h' \succsim_c h$ if and only if $h' = h$. In the foster care setting, it is important to note that children do not represent themselves in the matching process. Instead, caseworkers typically decide the best home for a child. Nevertheless, we will primarily refer to children as holding preferences and making decisions in our theoretical work for exposition's sake. When we highlight our applications, it will be an important fact that caseworkers are the real decision-makers.

In addition to these preferences, the matchmaker has evaluations for a child $c$ which is a ranking $>_c\hspace{5pt} \in P$ over prospective homes for that child. We assume that these evaluation rankings are strict. While evaluations will play a similar role to priorities, for example, in the school choice literature, they are not the same\footnote{They differ in an important aspect. The matchmaker's outcome need not be comparable across children even if it is comparable within a single child. Hence, it might not be possible to translate outcomes into priorities that homes have over children.}.

The vector of preferences for all children is $\succ_C \hspace{5pt} = (\succ_c)_{c \in C}$ and for all matchmaker evaluations is $> \hspace{5pt} = (>_c)_{c \in C}$. A preference double is a combination of these two: $\succ \hspace{5pt} = (\succ_C, >)$. A problem $L = (M, \succ)$ consists of the market, preferences, and evaluations. A mechanism $\phi[L]$ maps a problem to a feasible matching $\mu$ on $M$. From hereon, we omit "feasible" when referring to matchings and consider it implicit.

\subsection*{2.1 \hspace{5pt} Mechanism Desideratum}

A matching $\mu$ is $\succ$-Pareto dominated by $\mu'$ if, for all $c \in C$, $\mu'(c) \succsim_c \mu(c)$ and, for some $c' \in C$, $\mu'(c') \succ_{c'} \mu(c')$. A matching $\mu$ is $\succ$-\textit{efficient} if there does not exist another matching $\mu'$ that $\succ$-Pareto dominates it. We will generally refer to $\succ$-Pareto dominated and $\succ$-efficient simply as Pareto dominated and efficient in exposition. However, these concepts can apply to the evaluation order $>$ as well. In our experiment, we will utilize $>$-efficiency. Our interest is not only in efficient mechanisms. We also desire mechanisms that are robust in that they use the qualitative information that children's preferences encode while also heeding the algorithm evaluations. We say that a child $c$ has an unanimous improvement from $h$ if there exists some $h'$ such that $h \succ_c h'$ and $h >_c h'$. We write the set of unanimous improvers at $c,h$ as:
\[I(c,h) = \{ h' : h' \succ_c h \text{ and } h' >_c h' \}\]
$h$ is unanimously acceptable if $h \succ_c \varnothing$ and $h >_c \varnothing$. $c$'s not unanimously disagreeable or, for short, \textit{unanimous} homes are:
\[H^*(c) = \{ h \in H : h \succ_c \varnothing, h>_c \varnothing, \text{ and } I(c,h) = \varnothing\}\]
and $c$'s non-unanimous homes are:
\[H^-(c) = \{h \in H : h \notin H^*(c), h \succ_c \varnothing, \text{ and } h >_c \varnothing\}\]
Let the set of unanimous children on $\mu$ be:
\[I^*(\mu) = \{ c \in C : \mu(c) \in H^*(c) \}\] 
The set of acceptably matched children under $\mu$ is:
\[\kappa(\mu) = \{ c \in C : \mu(c) \in H^*(c) \cup H^-(c) \}\]
A matching $\mu$ is acceptable if there does not exist a $c$ such that $\mu(c) \notin H^*(c) \cup H^-(c)$. We restrict attention to acceptable matchings throughout this paper. $\mu$ is unanimously dominated by $\mu'$ if there is some coalition of children that unanimously improve under $\mu'$, none that are unanimously worse, and no matches are destroyed. Formally, $\mu$ is unanimously dominated by $\mu'$ if $I^*(\mu) \subset I^*(\mu')$ and $\kappa(\mu) \subseteq \kappa(\mu')$. We also say that $\mu$ is unanimously dominated by $\mu'$ if the coalition of unanimous children weakly expands and $\mu'$ matches strictly more children, i.e., $I^*(\mu) \subseteq I^*(\mu')$ and $\kappa(\mu) \subset \kappa(\mu')$. A matching $\mu$ is \textit{unanimous} if it is not unanimously dominated by any matching $\mu'$. We write $\mu \vartriangleright_{PD} \mu'$ and $\mu \vartriangleright_{UD} \mu'$ to say that $\mu$ is Pareto and unanimously dominated by $\mu'$, respectively. 

An efficient and unanimous matching is not guaranteed to exist (Proposition \ref{proposition:efficient-unanimous-existence}). We say that $\mu$ is \textit{constrained-efficient} if there does not exist another $\mu'$ such that $\mu \vartriangleright_{PD} \mu'$, $I^*(\mu) \subseteq I^*(\mu')$, and $\kappa(\mu) \subseteq \kappa(\mu')$. A constrained-efficient matching is not Pareto dominated by a matching that weakly expands the set of unanimous and matched children.

Unanimity lexicographically prioritizes full elimination of unanimous improvements whenever possible. However, when not possible, i.e. for any $c \in \kappa(\mu) - H^*(\mu)$, it imposes no restrictions. Constrained efficiency acts as a second-tier objective. We are also interested in \textit{partial} elimination of unanimous improvements whenever possible. A matching $\mu$ is improvable if there is another matching $\mu'$ such that, for every $c$, $I(c, \mu'(c)) \subseteq I(c, \mu(c))$, and $\mu'(c') \in I(c', \mu(c'))$ for at least one $c'$. $\mu$ is unimprovable if it is not improvable.

Let $\succ_c$ be $c$'s true preferences for some problem $L$ where $c$ reports $\succ_c$, and let $\succ'_c$ be a false preference report for child $c$. Denote $L'$ as the problem where $c$ reports $\succ'_c$ and all other reports are as in $L$. $\phi$ is \textit{strategy-proof} if no child can benefit by misrepresenting her preferences holding other reports fixed. Formally, $\phi$ is a strategy proof mechanism if, for all $c$, given any $L$ and $L'$ defined as above, we have that $\phi[L](c) \succsim_c \phi[L'](c)$.

\subsection*{2.2 \hspace{5pt} Discussion: Algorithms in Child Welfare}

Our framework captures the tension inherent in using algorithms to facilitate matching whenever humans are involved in the process. Humans can fundamentally disagree---misalign---with algorithms. We refer to misalignment as humans and algorithms having  either separate goals or separate ideals of pursuing the same goal. The first case is more fundamental and likely to be antithetical to unanimity. For example, if humans pursued some goal orthogonal to social welfare, then we can arrive in a case where evaluations are preferences in reverse order. Therefore, there would be no possible unanimous improvement and every matching would be unanimous. Redressing this requires a fundamentally different approach (\qcitenp{biro-2020}). We view this as a rare situation. The second case is more usual: humans and algorithms both seek to maximize social welfare, but they use different criteria for a "good" match, thereby arriving at different rankings over objects.

We demonstrate that it is feasible to combine human and algorithm decision-making to produce matchings. Our formulation of unanimity emphasizes prioritizing matches that either the human prefers or the algorithm prefers; we avoid matches where neither are true. Unanimity provides both provides robustness against utilizing only one decision model (an algorithm or a human) and synthesizes the information that algorithms and humans generate regarding the value of any particular match when the two parties seek the same goal. Mostly positive concerns motivate the first point: an algorithm could be poor precisely in special circumstances where a human sees an obvious error. Institutional concerns (often concerns in other environments as well) motivate our second point: with the future of a child on the line, many practitioners are unwilling to purely rely on an algorithm---even if it is accurate---to place children in homes for practical, ethical, and legal concerns. 

Unimprovable matchings are even more fair; any reallocation starting from an unimprovable matching cannot "weakly improve" all children if the reallocation unanimously improves any child. The basic rationale is that "weakly improving" should include an improvement with respect to at least one order \textit{that does not introduce a new unanimous improver}. The reallocation is not fair if it creates any "unanimous objection" that did not exist prior to the reallocation. Subject to this, unimprovable matchings fully prioritize unanimously improving children over improvements with respect to only one order. We show that strengthening this concept any further leads to non-existence results in the intuition preceding Proposition \ref{proposition:extended-unanimity}.

Our main task is showing that our criterion are beneficial in theory and practice. Theoretically, we show that unanimous matchings always exist, demonstrate how to construct them, and analyze their strategic and normative properties. Proposition \ref{proposition:efficient-unanimous-existence} highlights a tradeoff between unanimity and efficiency. Constrained-efficiency circumvents this tradeoff through improving on unanimous matchings. Empirically, we compare human judgment to our mechanisms using experimental data from real caseworkers' preferences.

Concretely, in foster care, a local matchmaker could use placement stability as its welfare measure, or evaluation, for children. The matchmaker's evaluations would be predicted placement stability. This is the probability that, if $c$ is placed at $h$, the placement would persist until $c$ permanently exits foster care. We will refer to this as \textit{persistence} to avoid confusion with the matching theory term \textit{stability}. Persistence is a generally agreeable goal among practitioners, but caseworkers still might fail to make the welfare optimal decision on behalf of the children they represent. Currently, caseworkers in almost all counties in the United States effectively follow a decentralized approach where the earliest caseworker to find a suitable home for a child can contact that home to begin the match. This is equivalent to a serial dictatorship ordered by time which will always yield a Pareto efficient matching. Using this process, one reason why caseworkers make suboptimal decisions is simple inexperience. \qcite{edwards-2018} estimates that the average caseworker only spends 1.8 years on the job before quitting. The field is notorious for high turnover rates, where the same authors estimate that 14-22\% of caseworkers turnover every year in the US. An inexperienced caseworker's judgment can, in turn, be quite poor, and the matchmaker should not rely on them alone for assessing persistence. This same concept applies to any organizational setting with misalignment.

One could hypothesize that the government could enforce an algorithm that matches children to homes purely based on its evaluations rather than caseworker preferences. However, to most practitioners, many features of foster care render it unamenable to algorithmic decision-making. Caseworkers routinely use qualitative information that is not easily encoded into usable data to make matching decisions. Moreover, as aforementioned, in a high-stakes environment, turning to full automation might not be an option. Last, caseworkers often simply do not trust algorithms. This motivates our use of preferences and evaluations. Including relevant caseworker datum and their choices in the process alleviates all of these problems. As it stands, to reap the boundless improvements from algorithmic decision-making seen in traditional matching theory applications like school choice, hospital-doctor matching, and kidney exchange, practitioners must agree to utilizing the methods invented in matching theory. Ensuring that they can incorporate algorithmic tools with human decision-making is a binding constraint. As we will detail in our empirical section, this strategy could greatly improve persistence in foster care and substantially benefit the tens of thousands of children entering foster care every year in the United States.

Finally, implementing algorithmic matching could constitute a partial solution to the major labor supply issues seen in child welfare. Caseworkers are notoriously overwhelmed with large caseloads, and a system that uses an objective metric to place children with the best possible home could simplify caseworkers' decisions, decrease the amount of manual labor involved, and increase retention. Importantly, the algorithms we devise cannot \textit{replace} caseworkers; we do not desire this. Instead, caseworkers would be an important check-and-balance against an algorithm that may have systemic prediction errors that could harm children's persistence. Thus, our mechanisms would achieve the best of both worlds: maximum utilization of qualitative data present in caseworker preferences and automated optimization of child outcomes using predictive analytics.

\section*{3 \hspace{5pt} Results}

In this section, we elucidate the structure of unanimity, detail our mechanisms for computing efficient and unanimous matchings, and give important results on their properties.

\subsection*{3.1 \hspace{5pt} The Structure of Unanimous Matchings}

We first begin with core properties of unanimous domination. Denote $U$ as the set of all unanimously acceptable and feasible matchings for the problem $L$.

\begin{lemma}\label{lemma:partial-order}
    The relation $\vartriangleright_{UD}$ is a strict partial order on the set $U$.
\end{lemma}

\begin{proof}
    We demonstrate irreflexivity, asymmetry, and transitivity. First, observe that for any $\mu = \mu'$, $I^*(\mu) = I^*(\mu')$ and $\kappa(\mu) = \kappa(\mu')$. Therefore, $\neg (\mu \vartriangleright_{UD} \mu')$. Second, if $\mu \vartriangleright_{UD} \mu'$ then $I^*(\mu) \subset I^*(\mu')$ or $\kappa(\mu) \subset \kappa(\mu')$, implying that $\neg (\mu' \vartriangleright_{UD} \mu)$. It remains to show transitivity. Suppose, for a contradiction, that $\vartriangleright_{UD}$ is not transitive on comparable elements in $L$. Then, there must exist some $x, y, z \in U$ such that $x \vartriangleright_{UD} y$, $y \vartriangleright_{UD} z$, but $\neg(z \vartriangleright_{UD} x)$. This implies that either (a) $I^*(z) \subseteq I^*(x)$ and $\kappa(z) \subseteq \kappa(x)$, (b) $I^*(z) \subset I^*(x)$, or (c) $\kappa(z) \subset \kappa(x)$. However, by definition of unanimity, we have that (i) either $I^*(x) \subset I^*(y)$ or $\kappa(x) \subset \kappa(y)$, (ii) either $I^*(y) \subset I^*(z)$ or $\kappa(y) \subset \kappa(z)$, and (iii) $I^*(x) \subseteq I^*(y) \subseteq I^*(z)$ and $\kappa(x) \subseteq \kappa(y) \subseteq \kappa(z)$. (iii) rules out case (b) and (c). By (i) and (ii) we have:
    \[I^*(x) \subset I^*(y) \subseteq I^*(z)\]
    or
    \[\kappa(x) \subset \kappa(y) \subseteq \kappa(z)\]
    which contradicts the remaining case (a). This proves the lemma.
\end{proof}

\begin{corollary}\label{corollary:unanimous-existence}
    A unanimous matching always exists.
\end{corollary}

\begin{proof}
    Suppose not. Then for any $\mu \in U$, there exists another $\mu'$ such that $\mu \vartriangleright_{UD} \mu'$. By transitivity of the partial order and finiteness of $U$, this is impossible.
\end{proof}

These two results show that unanimous matchings have well-ordered structure that trivially offers existence. The unanimous domination order is the same as Pareto dominance: irreflexivity, asymmetry, and transitivity all hold for both on the set of feasible matchings. Still, in the Proposition below, we find that this does not guarantee that combining unanimity and efficiency yields positive results. For the proof below, we denote the sets of unanimous homes (homes where the child has no unanimous improvement) as $H^* : C \rightarrow 2^H$.

\begin{proposition}\label{proposition:efficient-unanimous-existence}
    An efficient and unanimous matching can fail to exist.
\end{proposition}

\begin{proof}
    Consider the following market with $C = \{ a, b, c, d \}$ and $H = \{ 1, 2, 3, 4 \}$. The preferences and evaluations are as follows:
    \begin{samepage}
    \begin{multicols}{3}
        \noindent\begin{align*}
            &a\\
            \succ_a: 1, 2&, 3, 4\\
            >_a: 4,3&,1,2
        \end{align*}\columnbreak
        \begin{align*}
            &b\\
            \succ_b: 1, 3&, 2, 4\\
            >_b: 2,3&,1,4
        \end{align*}\columnbreak
        \begin{align*}
            &c\\
            \succ_c: 1,2&,3,4\\
            >_c: 1,2&,3,4
        \end{align*}\columnbreak
    \end{multicols}
    \begin{multicols}{3}
        \noindent\begin{align*}
            &d\\
            \succ_c: 4,3&,2,1\\
            >_c: 4,3&,2,1
        \end{align*}\columnbreak
    \end{multicols}\vspace{-20pt}
    \end{samepage}
    The sets of unanimous homes for each child are: $H^*(a) = \{ 1, 3, 4 \}$, $H^*(b) = \{ 1, 2, 3 \}$, $H^*(c) = \{ 1 \}$, and $H^*(d) = \{ 4 \}$. We can build the matching $\mu = (3, 2, 1, 4)$ which has unanimous children $I^*(\mu) = \{ a, b, c, d \}$. One can verify that this is the only matching where all children are unanimous. Therefore, it is the only unanimous matching. However, $\mu$ is Pareto dominated by $\mu' = (2, 3, 1, 4)$. This proves the Proposition.
\end{proof}

The key insight from this counterexample is that a child need not remain unanimous even after improving with respect to her own preferences. A Pareto improvement can entail a worse evaluation. In the previous Proposition, $a$ Pareto improving to $2$ from $3$ gives her the worst possible evaluation. Compared to this, the matchmaker would prefer her at $1$ and $a$ would prefer the same. This "non-monotonicity" of unanimity with respect to Pareto improvements implies the non-existence result. 

Next, we turn to the first steps to construct unanimous matchings. We will prove a useful proposition below. We introduce the notion of \textit{unanimous rankings}. We call $>_{I(c)} \in P$ the unanimous ranking for $c$ if it satisfies, for all $h,h' \in H$, $h >_{I(c)} h'$ if and only if either (a) $h \in H^*(c)$ and $h' \notin H^*(c)$ or (b) $h \in H^-(c)$ and $h' \notin H^*(c) \cup H^-(c)$. Otherwise, $h \sim_{I(c)} h'$, i.e., $h$ and $h'$ are indifferent on the unanimous ranking. For all $h$, $h >_{I(c)} \varnothing$ if and only if $h$ is unanimously acceptable. The vector of all unanimous rankings is $>_I\hspace{3pt}= \{>_{I(c)}\}_{c \in C}$. Under this ranking, the child is indifferent between all homes that have no unanimous improvement, and the child strictly prefers all homes with no unanimous improvement to homes that have unanimous improvements. 

\begin{proposition}\label{proposition:unanimous-reduction}
    A matching $\mu$ is unanimous on $L$ if and only if it is efficient on $L' = (M, >_I, >)$.
\end{proposition}

\begin{proof}
    In the Appendix.
\end{proof}

Proposition \ref{proposition:unanimous-reduction} identifies a striking reduction: instead of constructing a new mechanism to identify unanimous matchings, the problem can be simplified by constructing "unanimous preferences" that input into any Pareto efficient mechanism. The details are delicate because the unanimous ranking contains indifference. We describe the Serial Dictatorship with Indifference (SDI) mechanism that adequately accommodates this slight nuance in the next section. Using SDI and Proposition \ref{proposition:unanimous-reduction}, we obtain unanimous matchings. Nevertheless, the space of unanimous matchings can be large. Using the groundwork laid here, our main results offer two novel methods: one to find a constrained-efficient and unanimous matching and another to find incontestable matchings.

\subsection*{3.2 \hspace{5pt} Computing Efficient and Unanimous Matchings}

The mechanism we present in this section relies on modifications of the classic Top Trading Cycles (TTC) mechanism. Before describing these, we provide a formal definition for SDI which computes unanimous matchings. For short, we will denote the problem $L^I = (M, >_I, >)$ as the transformation of the problem $L = (M, \succ, >)$ with the unanimity order $>_I$ defined on $\succ$ and $>$.

\begin{figure}[t]\begin{singleframedindent}
    \begin{center}
        \textbf{Algorithm 1. Serial Dictatorship with Indifference}
    \end{center}

    Initialize the following: the set of chosen matchings $U^n$ such that $U^1 = U$.

    \textbf{Round $n \geq 1$:} 

    According to an arbitrary order, select the next child $c$ that has not been a dictator. 
    
    $U_c = \{ \mu \in U^n : \mu(c) \succsim_c \mu'(c) \text{ for all } \mu' \in U^n \text{ and } \mu(c) \succ_c \varnothing\}$, i.e., $U_c$ is the set of acceptable allocations that $c$ most prefers among available, acceptable allocations at round $n$. Set $U^{n+1} = U_c$. 
    
    If all children have been dictators or all homes are matched in every allocation in $U^{n+1}$, then terminate and return any $\mu \in U^{n+1}$. Otherwise, continue to round $n+1$.
\end{singleframedindent}\end{figure}

We provide Serial Dictatorship with Indifferences, an algorithm to compute unanimous matchings directly, for completeness. It is frequently alluded to in the literature on one-sided matching but rarely described in detail. SDI allows children to select allocations, or outcomes as referred to in social choice theory, that they are indifferent between. The next dictator selects allocations from those that remain, and so on. This process is Pareto efficient just as SD is in the setting with no indifference.

\begin{proposition}\label{proposition:sdi-unanimous}
    Serial Dictatorship with Indifference is unanimous.
\end{proposition}

\begin{proof}
    This follows from the fact that SDI is Pareto efficient. Suppose, for a contradiction, that it is not, i.e., $\phi_{SDI} \vartriangleright_{PD} \mu'$ for some $\mu' \in U$. Consider the earliest dictator $c$ in round $n$ such that $\mu'(c) \succ_c \phi_{SDI}(c)$. $c$ would have chosen $\mu'$ if it were available, and every allocation in every future round $j > n$ must keep $c$ indifferent to $\mu'$. $\mu'$ must not have been available in round $n$. Hence, some dictator $c'$ in an earlier round $k < n$ did not choose $\mu'$. This implies that $\phi_{SDI}(c') \succ_{c'} \mu'(c')$, contradicting $\phi_{SDI} \vartriangleright_{PD} \mu'$. Therefore, SDI is Pareto efficient. By Proposition \ref{proposition:unanimous-reduction}, it is unanimous when applied to the problem $L^I$.
\end{proof}

Unanimous matchings can exhibit large degrees of inefficiency; in the example given below, we demonstrate an extreme case where any matching is unanimous but only one is efficient. 

\begin{example}\label{example:unanimous-inefficient}
    \textit{Unanimous matchings can be inefficient.}

    Consider the following market with $C = \{ a, b, c \}$ and $H = \{ 1, 2, 3 \}$. The preferences and evaluations are as follows:
    \begin{samepage}
    \begin{multicols}{3}
        \noindent\begin{align*}
            &a\\
            \succ_a: 1, &2, 3\\
            >_a: 3, &2, 1
        \end{align*}\columnbreak
        \begin{align*}
            &b\\
            \succ_b: 2, 1&, 3\\
            >_b: 3, 1&, 2
        \end{align*}\columnbreak
        \begin{align*}
            &c\\
            \succ_c: 3,2&,1\\
            >_c: 1,2&,3
        \end{align*}\columnbreak
    \end{multicols}\vspace{-20pt}
    \end{samepage}
    All homes are unanimous for all children, but the only efficient matching has all children receiving their top choices, i.e., $\mu = (1, 2, 3)$.
\end{example}

We desire mechanisms that satisfy unanimity (or nearly so) and, subject to that constraint, are also efficient. Our first mechanism is a variant of Top Trading Cycles which uses the insight from Proposition \ref{proposition:unanimous-reduction} to improve on any unanimous matching so that it meets our constrained efficiency criteria. We describe the operation of our mechanism, Unanimous Top Trading Cycles (UTTC), in Algorithm 2.

\begin{figure}[t]\begin{singleframedindent}
    \begin{center}
        \textbf{Algorithm 2. Unanimous Top Trading Cycles}
    \end{center}

    Initialize the following: a unanimous matching $\mu^1$.

    \textbf{Round $n \geq 1$:} 

    Each child $c$ that has not been removed points to the child (or $\varnothing$) holding her most preferred remaining home in $H^*(c)$ if $\mu^1(c) \in H^*(c)$ or her most preferred, unanimously acceptable remaining home in $H$ otherwise. 
    
    There must be at least one cycle or chain (ending in $\varnothing$), possibly of length one. Set $\mu^{n+1}$ by matching each child to the next child's home in the cycle or chain; remove all involved children and homes.
    
    If no children remain, terminate and output $\mu^{n+1}$. Otherwise, continue to $n+1$.
\end{singleframedindent}\end{figure}

\begin{proposition}\label{proposition:uttc-efficient-unanimous}
    Unanimous Top Trading Cycles is constrained-efficient and unanimous.
\end{proposition}

\begin{proof}
    In the Appendix.
\end{proof}

Constrained efficiency requires that one cannot Pareto improve children without unanimously worsening a child. Our approach fixes a unanimous matching and uses TTC on a restricted problem where unanimous children only point to their favorite unanimous homes and animus children point to their unrestricted favorite home. By the unanimity of the original matching, it is a given that we cannot find another matching where we can expand the set of unanimous children. Therefore, we know that any matching that is more efficient requires unanimously worsening some child, and we can use this pointing rule without loss. The remainder of the proof for constrained efficiency involves showing that the rule still guarantees no Pareto improvements unless the improvement necessitates unanimously worsening a child, and this follows from the same reasoning as above. UTTC allows the matchmaker to improve on any unanimous matching which may be useful in large markets where a large number of homes can create substantial opportunity for Pareto improving trades. 

The main disadvantage of unanimity is that it virtually ignores children that are not matched to unanimous homes. As we discussed, UTTC addresses this partially---we resort to efficiency when unanimity cannot be guaranteed. Empirically, we show that this is an effective strategy. Nevertheless, it is possible for some coalition of children to have unanimous improvers when a reallocation could weakly decrease the set of unanimous improvers for all children. This implies that some unfairness can still exist in unanimous matchings. We use unimprovable matchings to fully address this fairness concern. Starting from an unimprovable matching, there is no reallocation that unanjmously improves at least one child without giving another child a unanimous improver that they did not already have. Our mechanism---Adaptive SDI (ASDI)---shows how our theoretical developments thus far yield a method to construct unimprovable matchings.

\begin{figure}[t]\begin{singleframedindent}
    \begin{center}
        \textbf{Algorithm 3. Adaptive Serial Dictatorship with Indifference}
    \end{center}

    Initialize the following: the set of chosen matchings $U^n$ such that $U^1 = U$ and remaining homes $H^1 = H$.

    \textbf{Round $n \geq 1$:} 

    According to an arbitrary order, select the next child $c$ that has not been a dictator. 

    \textit{Subroutine A}

    Construct the unanimous ranking for $c$ over remaining homes $H^n$. Denote the round $n$ unanimous ranking for $c$ as $>_{I(c)}^n$.

    \textit{Subroutine B}
    
    Set $U_c = \{ \mu \in U^n : \mu(c) \geq^n_{I(c)} \mu'(c) \text{ for all } \mu' \in U^n \text{ and } \mu(c) >^n_{I(c)} \varnothing\}$. Set $U^{n+1} = U_c$ and $H^{n+1} = \{h : \mu(h) = \varnothing \text{ for some } \mu \in U^{n+1}\}$
    
    If all children have been dictators or $H^{n+1} = \varnothing$, then terminate and return any $\mu \in U^{n+1}$. Otherwise, continue to round $n+1$.
\end{singleframedindent}\end{figure}

\begin{theorem}\label{theorem:idsdi}
    Adaptive Serial Dictatorship with Indifferences is unimprovable.
\end{theorem}

\begin{proof}
    In the Appendix.
\end{proof}

Our proof shows that the same logic in Proposition \ref{proposition:sdi-unanimous} holds with a few technical points to smooth out. Surprisingly, the adaptation (A) subroutine of ASDI neatly amends SDI to achieve unimprovability. The structure of ASDI reveals that unimprovable matchings are the "fixed point" of unanimous matchings when one allows for changing preferences in response to unavailable homes. Although this structure is evident, we remark that unimprovable matchings can fail to be unanimous. For example, one can run ASDI with the order $d, c, a, b$ in the example from Proposition \ref{proposition:efficient-unanimous-existence}. The result is that $\phi_{ASDI}(a) \in \{ 2, 3 \}$ and $\phi_{ASDI}(b) \in \{ 2, 3 \}$ whereas the only unanimous matching has $\mu(a) = 3$ and $\mu(b) = 2$. Although unimprovable matchings can fail to be unanimous, we observe that unimprovable matchings frequently coincide with unanimous matchings in the next section as we explore strategic incentives. 

\subsection*{3.3 \hspace{5pt} Strategic Incentives}

In this section, we prove several key results about the strategic properties of unanimous matching mechanisms. We find that there is no unanimous and strategy proof matching mechanism (Proposition \ref{proposition:unanimous-sp}). In addition, the manipulations that caseworkers could pursue are concerning to the healthy operation of any unanimous matching mechanism. Subsequently, based on \qcite{troyan-2019}, we evaluate our mechanism's \textit{obvious} manipulability, that is, whether or not an agent that has limited contingent reasoning can evaluate successful manipulations. All results hold for both unanimous and unimprovable matchings because the two concepts are equivalent for every example given.

\begin{proposition}\label{proposition:unanimous-sp}
    A unanimous and strategy-proof mechanism does not exist.
\end{proposition}

\begin{proof}
    Consider the following market with $C = \{ a, b \}$ and $H = \{ 1, 2, 3 \}$. The preferences, evaluations, and induced unanimity orders are as follows, where omitted homes are not unanimously acceptable and brackets indicate indifference:
    \begin{samepage}
    \begin{multicols}{2}
        \noindent\begin{align*}
            &a\\
            \succ_a: 1&, 2, 3\\
            >_a: 3&, 1, 2\\
            >_{I(a)}: [1&, 3], 2
        \end{align*}\columnbreak
        \begin{align*}
            &b\\
            \succ_b: 1&, 2, 3\\
            >_b: 3&, 1, 2\\
            >_{I(b)}: [1&, 3], 2
        \end{align*}\columnbreak
    \end{multicols}\vspace{-20pt}
    \end{samepage}
    The following matchings are unanimous: $(1, 3)$ and $(3, 1)$. A unanimous mechanism must output one of these matches. We prove that each one would induce some child to have a profitable manipulation.
    
    First, consider $(1, 3)$. We will denote $\succ'_b$ as the preference ranking where $b$ reports only $2$ as acceptable. Under this report, only $2$ would be a unanimous home for $b$. The resulting two unanimous matchings would be $(1, 2)$ and $(3, 2)$. In either case, $b$ has a profitable manipulation.

    Next, consider $(3, 1)$. Let $\succ'_a$ be the preference ranking where $a$ reports only $2$ as acceptable. Then the only unanimous matchings are: $(2, 1)$ and $(2, 3)$. In either case, $a$ has a profitable manipulation.

    Therefore, under any unanimous matching that is selected, either $a$ or $b$ have a profitable manipulation, contradicting the existence of a unanimous and strategy-proof mechanism.
\end{proof}
$a$ and $b$'s improving manipulations are simple: refuse to report less-preferred options when anticipating that they might be relegated to one of them. The Proposition trivially implies that there is also no unanimous and strategy proof mechanism since every unanimous matching is weakly unanimous. Unfortunately, this counterexample also demonstrates that any unanimous mechanism is manipulable by truncations in such a way that a falsifying child can find a Pareto improvement without necessarily matching to another unanimous home. In the worst case scenario where all agents adopt truncation strategies, this could severely harm children's welfare.

The very structure of unanimous matchings requires the matchmaker to use the reported preferences against the preferred interests of the children when there is no unanimous agreement. In particular, when the matchmaker knows that children have acceptable homes that generate no unanimous improvement, the matchmaker is obligated to attempt to assign the children to these homes irregardless of whether they are the children's most preferred homes. This proves to defeat strategy-proofness altogether. 

However, this all comes with one important caveat that has been an observation leading to a burgeoning literature of matching under incomplete information: it is unlikely that caseworkers know the preferences of other agents. In our setting---and in the case of unanimous matching in general---the asymmetric information is more extreme, because caseworkers do not know the matchmaker's evaluations for all children\footnote{The evaluations for a given child might be reported to the caseworker that represents that child, but they need not be reported to caseworkers that do not represent the child.}. How, then, should a caseworker evaluate manipulations under incomplete information? We adopt the notion of non-obvious manipulability as in \qcite{troyan-2019}. In short, a mechanism is non-obviously manipulable if there is no manipulation that makes a child strictly better off in the worst case or best case compared to truth-telling when she does not know others' preferences and evaluations.

We assume that every child knows her own preferences and evaluations. Fixing $M$, a matching $\mu$ is consistent with $\phi$, $\succ_c$, and $>_c$ if there is some profile of preferences $\succ_{-c}$ and evaluations $>_{-c}$ such that $\phi[M, \succ_c, \succ_{-c}, >_c, >_{-c}] = \mu$. We denote the set of matchings consistent with $\succ_c$ and $>_c$:
\[CU[\succ_c, >_c] = \{ \mu : \phi[M, \succ_c, \succ_{-c}, >_c, >_{-c}] = \mu \text{ for some } \succ_{-c},>_{-c} \}\]
The worst possible consistent matchings for $c$ given $\succ_c$ and $>_c$ are:
\[CU^-[\succ_c, >_c] = \{ \mu : \mu \in CU[\succ_c, >_c] \text{ and } \mu'(c) \succsim_c \mu(c) \text{ for all } \mu' \in CU[\succ_c]\}\]
and the best possible consistent matchings for $c$ given $\succ_c$ and $>_c$ are:
\[CU^-[\succ_c, >_c] = \{ \mu : \mu \in CU[\succ_c, >_c] \text{ and } \mu(c) \succsim_c \mu'(c) \text{ for all } \mu' \in CU[\succ_c]\}\]
We will sometimes refer to these worst and best consistent matchings as singular without loss of generality as we could select an arbitrary element from the sets. We say that $\succ'_c$ is a worst-case manipulation if:
\[CU^-[\succ'_c, >_c](c) \succ_c CU^-[\succ_c, >_c](c)\]
and $\succ'_c$ is a best-case manipulation if:
\[CU^+[\succ'_c, >_c](c) \succ_c CU^+[\succ_c, >_c](c)\]
A mechanism $\phi$ is non-obviously manipulable if no $c$ has a worst-case nor best-case manipulation. We can now state our result.
\begin{theorem}\label{theorem:obvious-manipulability}
    The following statements are true:
    \begin{enumerate}
        \item[i.] Any unanimous matching mechanism is obviously manipulable.
        \item[ii.] Serial Dictatorship with Indifference is non-obviously manipulable if $|H| \leq |C|$.
    \end{enumerate}
\end{theorem}

\begin{proof}
    In the Appendix.
\end{proof}

The Theorem is silent on UTTC. It is easy to prove that the mechanism is non-obviously manipulable conditional on the initial endowment. However, UTTC must always input a unanimous matching. Any unanimous matching mechanism it uses to generate an endowment will be obviously manipulable. Therefore, UTTC is obviously manipulable by proxy.

The proof for the non-obvious manipulability of SDI under the imbalance condition is fairly straightforward. The best-case for a child $c$ when truth-telling is always matching to her most preferred home which is always unanimous. We can find another set of preferences and evaluations for other children that will ensure that this match for $c$ is unanimous. Therefore, the key step is to prove that the worst-case for truth-telling is never strictly worse than the worst-case under any misreport. This follows from one observation. Children that go later in the order are at a disadvantage when children earlier in the order will permanently match to unanimous homes, which occurs when these children's unanimous homes align with their most preferred home. Under such a worst-case conjecture, manipulations can never help. If we return to the same example in Proposition \ref{proposition:unanimous-sp}, when $a$ reports that only $2$ is acceptable, she will be unmatched under the conjecture that either $b$ or $c$ is prior to her in the order and has a most preferred home $2$.

Condition (i) in Theorem \ref{theorem:obvious-manipulability} is bad news. When the number of homes is large relative to the number of children, the obvious strategy for a child is to truncate her preferences to remove her least preferred homes (up until the length of her preference list matches the number of children) so that there is no possibility that the mechanism matches her to one of them to enforce unanimity. However, condition (ii) provides reassurance. First, for our application, there is almost never a foster home surplus, meaning that the Theorem will apply. Second, the proof for condition (ii) shows that the mechanisms have no best-case manipulation, and the \textit{only} worst-case manipulation are truncations leaving at least $|C|$ homes as unanimously acceptable. We conclude that it is critical that the matchmaker is aware of the strategic incentives that might arise when implementing our mechanisms and that an organization should strive to align its decision-makers and algorithm with its objective in order to mitigate strategic behavior.

\subsection*{3.4 \hspace{5pt} A General Theory of Matching with Preference Aggregation}

Our results above offer a promising direction for future research. Our primary focus in this paper is the applied problem of matching foster children to foster homes with human and algorithm decision-making, but Proposition \ref{proposition:unanimous-reduction} hints that a number of preference aggregation rules other than unanimity could emerge as viable candidates for this task. Here, we briefly outline a general framework that we term matching with preference aggregation, and we show that no desirable aggregation rule can restore positive strategy-proofness results. An aggregation rule will input preferences and evaluations and output an aggregated ranking. We will then seek to find matchings that are efficient with respect to the aggregated ranking and mechanisms that are, in some sense, strategy-proof on the problem with the aggregated ranking. Strategy-proofness takes on a new form in this framework owing in part to institutional constraints and in part to the intricacies of matching with preference aggregation. We address each.

\textit{Group Efficiency Concerns in Practice}---As we outlined in our Introduction, no counties that we are aware of use matching mechanisms to match children to homes. However, there are some that use match recommendation tools. \qcite{slaugh-2016} redesigns the match recommendation system used by a Pennsylvania statewide child welfare authority. Children unable to find a match in the local county would be forwarded to the state authority, and the state would use the system to identify suitable homes for the child. As of 2025, we verified with the practitioners in question that the general procedure outlined here was still operational. However, an issue that Slaugh et al. identified---and that we confirmed still exists to some degree---is that the match recommendation system was not heeded because caseworkers questioned the matches it proposed. We quote their account of the problem highlighted when they surveyed the caseworkers in Pennsylvania counties: "[O]nly 32 percent
agreed with the statement that PAE does a good job of
recommending the most suitable families via electronic
matches from the Resource Family Registry for each
child."

We interviewed caseworkers in the county that parroted this sentiment. Another perspective from Slaugh et al. explains how they view the role of the new system they provided to the state: "[O]ur algorithm generates a list of mere recommendations that may be implemented in conjunction with the judgment of the professionals ... [In foster care], such ranking information is impossible to elicit directly because of limitations such as market size and informational asymmetry. One main function of our algorithm can be seen as constructing such preferences from given pieces of information [about homes] and using them as the basis of recommendations." In other words, caseworkers cannot fully form preferences, so the algorithm completes them by using a large swathe of information about homes that humans would otherwise need to manually sort through. Then, the caseworker reviews the recommendations and selects matches that accord with her own, true preferences. One can see this as an informal case of matching with preference aggregation where the recommendation system is the algorithm; the key difference is that the matches it produces are not binding.

The intuition behind why a county might abandon this match recommendation system is \textit{not} that a single caseworker might find matches she prefers more than the recommendation. Consider a situation where one caseworker is unsatisfied, others receive their most preferred match from the recommendation system, and there are no unmatched homes. This is Pareto efficient. The fundamental problem is that multiple caseworkers could be unsatisfied with recommendations---being able to jointly improve from a "recommendation reallocation"---or that unmatched homes might be more preferable than the recommendation. In other words, matching with preference aggregation does not always produce efficient allocations. One might hypothesize that, if the recommendations were efficient, even if there were caseworkers who could benefit by strategically providing incorrect information to the recommendation system to benefit their recommendation, they would not do so. The reason why is that caseworkers are to do right by all children and colleagues, but strategic manipulation of an efficient matching is guaranteed to harm some other caseworker's preferences.

\textit{Theoretic Concerns}---This connects to the fine details of strategic incentives in matching with preference aggregation. As we demonstrated with Example \ref{example:unanimous-inefficient}, a matching that is efficient with respect to aggregated preferences can be $\succ$-inefficient. Therefore, our first criteria for a "strategy-proof" aggregation rule and mechanism is that there should be no profitable manipulation by a child that does not harm another child. This is a relaxed version of strategy-proofness that considers only particular types of manipulations. 

We add a second criteria. There should additionally be no profitable manipulation by the \textit{matchmaker} that can weakly improve all children with respect to its preferences. Suppose that this did not hold. A matchmaker that uses, for instance, predictive analytics to assess match persistence could have strange and perverse incentives to artificially distort the predictions. One could argue that this reduces confidence in utilizing a mechanism with aggregated preferences at all because the agents cannot be certain that the matchmaker is utilizing the true welfare predictions over strategic distortions. Such a flaw could induce serious legal and ethical concerns. The matchmaker might even be driven to dictate the matches purely by the algorithm---a very undesirable option as we discussed in the Preliminaries---thereby avoiding inefficiencies. Hence, we would like to ensure that the mechanism attains $>$-efficiency when the matchmaker provides the true algorithmic predictions to the mechanism.

We retain our previous notation. An aggregation rule $\tau : P \times P \rightarrow P$ maps two rankings (a preference and an evaluation) into one ranking. The aggregate ranking vector is $\bar{\tau} = \{\tau(\succ_c, >_c)\}_{c \in C}$. Every aggregation rule defines a subproblem $L_\tau \equiv (M, \bar{\tau})$. $\phi$ is ($\tau$-)group robust if (i) for all $c$ and manipulations $\succ'_c$, if $\phi[(M, \tau(\succ'_c, >_c), \bar{\tau}_{-c}](c) \succ_c \phi[L_\tau](c)$, then $\phi[L_\tau](c') \succ_{c'} \phi[(M, \tau(\succ'_c, >_c), \bar{\tau}_{-c}](c')$ for some $c' \neq c$ and (ii) for all $c$ and manipulations $>'_c$, if $\phi[(M, \tau(\succ_c, >'_c), \bar{\tau}_{-c}](c) >_c \phi[L_\tau](c)$, then $\phi[L_\tau](c') >_{c'} \phi[(M, \tau(\succ_c, >'_c), \bar{\tau}_{-c}](c')$ for some $c' \neq c$.

This framework has immediate parallels to social choice theory and, in anticipation for the main result in this section, the Gibbard-Satterthwaite theorem (\qcitenp{gibbard-1973}, henceforth GS). The aggregation $\tau(\succ_c, >_c)$ is a social outcome with two interested parties: the child $c$ and the matchmaker. This reveals three important distinctions between matching with preference aggregation and social choice theory. First, there are only two parties. Many "voting rules" (aggregations) will be degenerate due to frequent ties. Second, we greatly relax strategy-proofness to only outlaw manipulations that are efficient improvements with respect to either $\succ$ or $>$. Third, the social outcome is not important per-se. It is only relevant insofar as it affects matchings. These three differences might offer some hope that group robust mechanisms exist more in wider variety than in social choice theory. Our main result below shows that this is not true for most desirable rules and mechanisms.

Before we state the result, we give additional definitions. Let $T_c \equiv \tau(\succ_c, >_c)$ be the aggregated preference relation. We will consider aggregation rules $\tau$ that select strict rankings, noting that we will show that our mechanism desideratum from earlier sections can be reduced to a case of strict preference aggregation. $\tau$ follows the \textit{weak Pareto principle} (WPP) if, (a) for any $h,h' \in H$ such that $h \succ_c h'$ and $h >_c h'$, then $h T_c h'$ and (b) for any $h \in H$ such that $\varnothing \succ_c h$ and $\varnothing >_c h$, then $\varnothing T_c h$, and (c) for any $h \in H$. such that $h \succ_c \varnothing$ and $h >_c \varnothing$, then $h T_c \varnothing$. WPP requires that ranking a home higher on both rankings than another ensures that the home is at least as highly ranked under $\tau$ as the other. It also requires that unanimously acceptable homes remain acceptable and unanimously unacceptable homes remain unacceptable. 

A (deterministic) mechanism $\phi$ satisfies \textit{independence of irrelevant alternatives} (IIA) if it is independent of worse alternatives (IWA) and independent of unmatched alternatives (IUA). IWA requires a child shuffling homes without lowering the ranking of the child's matched home has no effect on a mechanism. Formally, whenever $\phi[L] = \mu$ for some problem $L$, $\phi[L'] = \mu$ for any problem $L' = (M, \succ')$ where $\succ'_{-c} = \succ_{-c}$ and $\succ'_c$ satisfies if $\mu(c) \succ_c h'$ ($\mu(c) \sim_c h'$) for some $h' \in H$, then $\mu(c) \succ'_c h'$ ($\mu(c) \succsim'_c h'$). A mechanism satisfies IUA if a child's misreport that causes her to be matched to an unmatched home does not affect other children if they would not match to the child's old home. Formally, for any $\succ,\succ'$ where $\succ'_{-c} = \succ_{-c}$, if $\phi[L'](\phi[L](c)) = \varnothing$ and $\phi[L](\phi[L'](c)) = \varnothing$, then $\phi[L](c') = \phi[L'](c')$ for any $c' \neq c$. It is immediate that SD will satisfy both criterion.

A serial choice rule (SCR) $\tau$ is defined by a predetermined series of dictators which we can denote $\tau^D = (c, m, ...)$ where $c$ indicates that the child is the dictator, and $m$ indicates that the matchmaker is the dictator. The aggregated preference is constructed by the dictator choosing her top home that has not been chosen before to be the next home in the ranking. The next dictator does the same and so on. We denote the length of the sequence $\tau^D$ as $\#\tau^D$ and assume that $\#\tau^D \geq |C|$. A SCR is dictatorial if every element in $\tau^D$ is equivalent, i.e., if $\tau = \succ_c$ or $\tau = >_c$ in the case where $\#\tau^D \geq |H|$.

\begin{theorem}\label{theorem:gs-aggregation}
    When $|H| \geq |C| + 2$, for any $\tau$ satisfying the weak Pareto principle, an efficient, strategy-proof, and IIA mechanism $\phi$ is group robust if and only if $\tau$ is a dictatorial serial choice rule.
\end{theorem}

\begin{proof}
    In the Appendix.
\end{proof}

Theorem \ref{theorem:gs-aggregation} extends GS to our setting with two relatively innocuous assumptions: WPP and IIA. The WPP is necessary to assess aggregation rules that are sensible. One example of a rule that does not follow the WPP but ensures group robustness is a rule that makes all homes indifferent for any preferences and evaluations. In our proof, IIA essentially requires the mechanism to have behave normally when decision-irrelevant preferences change or when a child is reallocated between homes that no other child prefers to her own match (i.e., an unmatched home in an efficient matching). Unfortunately, we cannot argue that the mechanism will remain consistent when a child's preferences change without this assumption, though it may be possible to prove the Theorem using a less restrictive version of the condition. 

The intuition for the first set of claims that prove that $\tau$ is a serial choice rule, is, frankly, rather opaque. The main thrust of the argument relies on using similar proofs for GS to show that in a single-agent environment, the matching is a social outcome, and, therefore, an efficient, IIA, and strategy-proof (equivalent to group-robust with one agent) mechanism will be dictatorial in the sense that either the child or the matchmaker will choose the matching. Even with this proven, extending the same idea to multi-agent markets is not easy because we must show that the SCR is dictatorial for multiple choices. We do this by finding a multi-agent problem where at least one agent faces the single-agent problem. We use induction to assume that the aggregation rule must be a SCR. Next, we can show that the same proof for the single-agent problem implies that the SCR maintains its dictatorial structure and that it can be extended to $n+1$. The Theorem additionally implies that, given the conditions, there is no $\tau$-efficient, strategy-proof, and IIA mechanism that is simultaneously $\succ$- and $>-$efficient. Such a mechanism could not have a group-improving manipulation, contradicting the Theorem. Furthermore, there is no $\tau$-efficient, strategy-proof, and IIA mechanism that prevents strategic manipulation of the aggregation rule on an individual basis because every group-improving manipulation rule is a individually-improving manipulation.

This result is sweeping and implies that unanimity's susceptibility to strategic incentives is not unique. The Theorem would apply if we strengthened group-robustness to an individual strategy-proofness because a group-beneficial manipulation is, of course, individually beneficial. We assumed that the matchmaker would not manipulate in Section 3.3, but the Theorem suggests that, in that case, avoiding manipulations by children could entail dictating according to their preferences. This point is not definitive, but we conjecture it to be true. Moreover, the Theorem guarantees that that any appealing (WPP) method to aggregate preferences in one-sided matching will fall prey to manipulations. Nevertheless, as it is in classical social choice theory, different rules may have varying degrees of manipulability and other positive properties. We close out this section with a survey of potential aggregation rules.

\begin{example}\label{example:regular-rule}
    All of the following rules satisfy the WPP.

    \textbf{Borda Count}
    
    Borda Count is a positional voting rule that assigns points to homes based on their rank in each order and re-ranks homes from least-to-most points. Let $Rank(h,\succ_c)$ be $h$'s rank in $\succ_c$ if $h$ is acceptable and some number $K$ otherwise. If $h$ is acceptable in $\succ_c$ and $>_c$, its score is $Score(c,h) = Rank(h, \succ_c) + Rank(h,>_c)$, then $h T^{BC}_c h'$ if $Score(c,h) > Score(c,h')$. Otherwise, $\varnothing T^{BC}_c h$. Break ties according to a predetermined, fixed rule.
    
    \textbf{Min Rank} 
    
    The Min Rank rule sorts homes by their best-case rank. Let $Rank(h,\succ_c)$ be $h$'s rank in $\succ_c$ if $h$ is acceptable and $\infty$ otherwise. Define $Min(h,c)$ as $\min_{\vartriangleright \in \{\succ_c, >_c \}} Rank(h, \vartriangleright)$. Under $\tau_{Min}$, $h T_{c}^{MIN} h'$ if $Min(h,c) < Min(h',c)$ if $Min(h,c)$ is finite. If $Min(h,c) = \infty$, then $\varnothing T_c^{MIN} h$. Break ties according to a predetermined, fixed rule.

    \textbf{Max Rank}
    
    The Max Rank rule is the mirror of Min Rank. It sorts homes by their worst-case rank. Let $Rank(h,\succ_c)$ be $h$'s rank in $\succ_c$ if $h$ is acceptable and $\infty$ otherwise. Define $Max(h,c)$ as $\max_{\vartriangleright \in \{\succ_c, >_c \}} Rank(h, \vartriangleright)$. Under $\tau_{MAX}$, $h T_{c}^{MAX} h'$ if $Max(h,c) < Max(h',c)$ if $Max(h,c)$ is finite. If $Max(h,c) = \infty$, then $\varnothing T_c^{MAX} h$. Break ties according to a predetermined, fixed rule.
\end{example}

In addition to the above examples, we examine unanimity-like aggregation rules. It could be objected unanimity does not satisfy WPP because non-unanimous homes are not strictly ranked when there is a unanimous improvement. Additionally, the aggregation is not strict. Therefore, it fails the standards we imposed in Theorem \ref{theorem:gs-aggregation}. This is true. Yet, our refinements address this. Let $\tau^u$ be the \textit{extended unanimity} rule where if $h \in H^*(c)$ and $h' \notin H^*(c)$, then $h T^u_c h'$. For any $h,h' \in H^*(c)$ or $H^-(c)$ such that $h \succ_c h'$, $h T^u_c h'$. For any $h \in H$, if $\varnothing \succ_c h$ or $\varnothing >_c h$, then $\varnothing T^u_c h$. The extended unanimity order "completes" unanimity by ordering homes that are in the same unanimity sets according to preferences. It is necessary to order by preferences (or evaluations) instead of unanimous improvements because the latter will not necessarily yield a total preorder\footnote{Consider the rankings $1, 2, 3, 4, 5, 6$ and $1, 2, 6, 4, 3, 5$. It must be that $3$ is strictly better than $5$, but $3$ and $6$ are indifferent as well as $5$ and $6$. To see how this translates into a failure to yield a "strong unanimity" solution concept, suppose that the market contains six children with identical preferences and evaluation as above. A matching is "strong unanimously dominated" if there is another matching where all children improve with respect to at least one ranking and at least one child improves with respect to both rankings. In any matching $\mu$, we can set $\mu'(\mu(3)) = 6$, $\mu'(\mu(6)) = 5$, and $\mu'(\mu(5)) = 3$. $\mu'$ strong unanimously dominates $\mu$. Therefore, a strong unanimous matching will fail to exist.}. We prove the following:

\begin{proposition}\label{proposition:extended-unanimity}
    The following statements are true:
    \begin{enumerate}
        \item[i.] $\tau^u$ is a strict total order,
        \item[ii.] $\tau^u$ satisfies WPP, and
        \item [iii.] If $\mu$ is constrained-efficient and unanimous, then $\mu$ is efficient under $\tau^u$.
    \end{enumerate}
\end{proposition}

\begin{proof}
    In the Appendix.
\end{proof}

The reverse of (iii) does not hold. Since unanimity requires that we maximize the unanimous coalition, using efficiency with respect to $\tau^u$ might prevent a mechanism from enlarging the set of unanimous children by virtue of matching children according to unanimity and preferences. Proposition \ref{proposition:extended-unanimity} shows that when efficiency and unanimity are in tandem goals, we cannot escape Theorem \ref{theorem:gs-aggregation}, and we are unlikely to find comparable aggregation rules that are more desirable. Proceeding by seeking rules that minimize opportunities for children or the matchmaker to improve---as $\tau^u$ does---is likely to be one of the most promising escape routes to our impossibility result.

\section*{4 \hspace{5pt} Experiment Design}

In this section, we assess the welfare implications of our mechanisms using experimental analysis. We detail our choice experiment design to elicit real caseworker preferences in a simulated market. Our simple experimental mechanism elicits subjects' ordinal preferences in an "elicited preferences" framework. We partition subjects into small markets with five children and five homes each, and we privately randomize each home to be either available or unavailable. Each subject receives a unique child assignment detailing her child's characteristics. The subjects then rank five homes. Our mechanism then matches each subject's child to her favorite available home. Subjects are incentivized to match their children to homes where the subject evaluates that the child would persist. The procedure is incentive-compatible and yields subjects' true preferences for the task of matching children to homes to maximize persistence. Our experiments' subjects are real caseworkers recruited through the research arm of a organization promoting foster care, allowing us to evaluate practitioners' behavior. In what follows, we provide the economic environment, preference parameters of interest, and experimental procedures. The experiment has not yet been run, so we do not present results.

\subsection*{4.1 \hspace{5pt} Economic Environment}

We setup a set of markets $M$ with a typical market $m \in M$. Each market $m$ has a set of children $C_m$ with $|C_m| = 5$ and a set of homes $H_m$ with $|H_m| = 5$. We do not specify caseworker preferences because our goal is to elicit their true preferences for a matching objective which we will specify shortly. We allow any matching between any child and any home in a market to be admissible.

The \textit{true persistence} for a child $c$ to a home $h$ is an unobserved outcome $Y(c,h) = 1$ if $c$ would persist at $h$ and $Y(c,h) = 0$ if not. The \textit{predicted persistence} is a classification predictor $\Hat{Y}(c,h) = 1$ if $c$ is classified to persist at $h$ and $\Hat{Y}(c,h) = 0$ if not. We detail the construction of this predictor in the next section. For our experiment, we assume that $Y(c,h) = \Hat{Y}(c,h)$ for all $c$ and $h$. We tell the subjects that the researcher has used a proprietary method to determine which matches will persist. Later, we conduct analysis allowing for misclassification error to understand its impact on child welfare. 

There are five rounds of the experiment. In each round, a caseworker (henceforth, subject) represents one distinct child. Additionally, each subject is randomly assigned to a new market $m$ in each round of the experiment. Note that this allows us to simulate matchings for one-hundred distinct markets such that $|M| = 100$. In a round, the objective for a subject is to match the child $c$ that she represents to a home $h$ so that $c$ persists. The subject may view a child profile for the child that she represents outlining various descriptive factors of the child and all observational covariates known to the researcher. The subject additionally may view home profiles for each home which present similar features. Each subject representing a child $c$ reports her preferences $\succ_c$ over $H_m$.

We use our mechanism, Random Serial Assignment (RSA), to determine the final matching. RSA inputs an \textit{availability function} $a : H \rightarrow \{ 0, 1 \}$ that determines whether or not \textit{any} child can be matched to a given home $h$. We use a uniform availability function where $a(h)$ is one with probability one-half and zero with probability one-half. Subsequently, RSA attempts to sequentially match each child to her highest ranked home. Importantly, the mechanism ignores feasibility constraints and may over-capacitate a home. With this caveat, RSA is strategy-proof (Proposition \ref{proposition:rsa-sp})\footnote{In fact, truth-telling is the unique non-trivial, weakly dominant strategy (Proposition \ref{proposition:rsa-unique}).}. Additionally, the match for any given child $c$ is independent from the reports of others. Thus, RSA allows us to measure subjects' preferences with a simple choice experiment with immediate payments rather than implementing a more complex interactive experiment. We can then use this preference data to simulate counterfactual average persistence under SD, SDI, and UTTC. Subjects' preferences input into RSA $\phi_{RSA}$ to yield the final matching. If $Y(c,\phi_{RSA}(c)) = 1$, then the subject receives a payoff of five experimental units (fifty cents). If not, she loses five experimental units. If the subject does not match to any home or times out, she loses two experimental units. Eliciting preferences with unacceptable homes requires a non-persisting match to yield lower payoff than no match. After all experimental rounds end, persistence for every round is revealed to the caseworker and payoffs realize (therefore, we provide no feedback at the end of an individual round). The key parameters of interest we desire from this experiment is the vector of caseworker preferences $\succ$. 

A meaningful analysis of the impacts of our mechanisms requires these \textit{real} preferences opposed to the induced preferences methods common in matching market experiments. \qcite{budish-2021}, studying MBA student preferences over course schedules, introduces the "elicited preferences" experimental methodology that broadly captures our mechanism RSA. However, their method involves no incentives and no formal mechanism. We iterate on this methodology through implementing a strategy-proof, incentivized elicited preferences framework.

\begin{figure}[t]\begin{singleframedindent}
    \begin{center}
        \textbf{Algorithm 4. Random Serial Assignment}
    \end{center}

    Input: an availability function $a : H \rightarrow \{ 0, 1\}$. Initialize: the matching $\mu^1$ such that $\mu^1(c) = \varnothing$ for all $c$ and $C^1 = C$.

    \textbf{Round $n \geq 1$:} 

    According to an arbitrary order, select the next $c \in C^n$ with $\mu^n(c) = \varnothing$. Select $c$'s most preferred home $h$ with $a(h) = 1$ if any exists. Set $\mu^{n+1}(c) = h$, if possible, and set $C^{n+1} = C^n - c$. 
    
    If $C^{n+1} = \varnothing$, terminate and return $\mu^{n+1}$. Otherwise, continue to round $n+1$.
\end{singleframedindent}\end{figure}
Using $\succ$, we can directly simulate counterfactual matchings with truthful reporting. Therefore, this removes the need for a treatment group where the final matching is determined through running our mechanisms. Nevertheless, this step precludes us from providing results regarding strategic manipulation. We propose that this would be an excellent area for future experimental work.

Our experiment elicits $\succ$ in the case that caseworkers are incentivized to act in the best interest of children's persistence. Misalignment between persistence and caseworker preferences arises only from caseworkers' erroneous (or advantageous) judgment. Therefore, our experiment captures the exact setting that we have envisioned in our Introduction and Preliminaries, namely, where decision-makers and the organization are aligned in goal but potentially misaligned in actions due to informational frictions.

\subsection*{4.2 \hspace{5pt} Profiles and Classification Predictor}

An important design choice in our experiment is the usage of a simulated market. However, this precludes us from assigning valid counterfactual persistence for a child-home match, because one cannot know which matches would persist in reality without conducting extensive, causal econometric analysis. Instead, we use \textit{predicted} persistence for our experiment.

We construct this classification predictor using real data. The data allows us to train a random forest predictor---$\Hat{Y}(c,h)$---to determine which matches will persist based on known, observable characteristics of children and foster homes. 

\textit{Data}---The Adoption and Foster Care Analysis and Reporting System (AFCARS) maintains a database of case-level on all children in foster care under entities (all states and many other agencies) that are mandated to report through Title IV-E. These datum include all children under eighteen (as of the end of the reporting year) in foster care in the United States from September 30th, 1999 to September 30th, 2021. Initially, this data set contains over fifteen million observations. A unique child identifier allows us to track children from year to year as long as they remain in foster care. Unfortunately, the AFCARS database does not maintain \textit{placement-level} information on all children. If a child moves from one foster home to another within-period, we only observe her final foster home at the end of the period. The econometrician can observe persistence for placements in which the child exits foster care or for the final placement in a year that does not remain the same in the next year\footnote{The missing placement observations may or may not be identically distributed to the observed placements that do not persist. Counties must report placements at the close of two biannual reporting periods. We observe the most recent placement for all children as of the close of the second biannual reporting period. Identical distribution in the observed and unobserved populations would fail if placements made before a child's most recent placement systematically differ. Nevertheless, as we do not aim for causal estimation, this is only important insofar as it affects the accuracy of our machine learning model. One can account for out-of-sample prediction bias using adjustment methods.}. Nevertheless, this restriction leaves us with a large data set as we discuss below. The AFCARS database contains vital information about foster children: demographics, needs and disabilities, reason for entry in foster care, and more. Foster home demographics and structure are also available for observed placements. 

We transform the AFCARS database into Placement Files. Our Placement Files data adds an indicator for persistence equal to one if a child has no subsequent placement during her stay in foster care and equal to zero otherwise. We remove duplicate placements (if a placement remains static from year to year) and drop placements for children whose placements remain the same from some year until September 30th, 2022 if the child has not exited foster care\footnote{We are aware that in classical results, this is not an optimal solution for right-censoring bias. However, the implications of how this might affect training a random forest are unclear and beyond the scope of our current analysis.}. We keep children in within-state placements that are in a non-relative foster home, group home, or institutional home that are not waiting for adoption\footnote{This is the representative case for a child requiring a temporary foster home match as we envision for our application.}. We reduce highly correlated covariates into a single more interpretable covariate (for example, child date of birth and age at first day of fiscal year combine to age at time of placement). We also remove several administration-related covariates that caseworkers cannot plausibly interpret to impact placement stability on the child-home level (state FIPS code, agency FIPS code, and others). We exclude covariates that are ex-post information (length of days in placement, actual financial payments made to foster homes, and others). Lastly, we drop observations with null values for child or non-relative foster home demographic columns. Although these observations may represent children or homes with unknown demographics\footnote{This is not always true because AFCARS does have specific fields to indicate "unknown".}, we choose to drop them because we want to compare caseworker and algorithmic decision-making will \textit{all} available information. We do not perform this step for congregate care foster homes because null values correspond to non-applicability rather than unknown information. Using these criteria, we retain a data set of $2,335,517$ unique placements from 2000 to 2022. The original list of covariates and the final list are in Appendix D.

Table \ref{table:placement-summary-stats} describes the demographic characteristics of children in care of the county and structural characteristics of foster homes. Male children are slightly over-represented. As was previously known, African American children are disproportionately represented in foster care---thirty-two percent are African American, sixty-two percent are White, sixteen percent are Hispanic, and four percent are Native American. One quarter have a clinically diagnosed disability. In placements, approximately forty-two percent take place in married couples' homes, three percent in unmarried couples' homes, twenty-one percent in single parents' homes, and a combined thirty-one percent are in congregate care settings\footnote{Group settings are congregate care homes with twelve or less children. Institutional homes are all other congregate care homes.}. 

We also report average persistence for the same categories in Table \ref{table:placement-summary-stats}. Female children have the same average rate of persistence as male children which is the sample mean 51\%. Children with disabilities have the lowest rate of persistence at 45\%, followed by Native American at 49\%, and Black and Hispanic children at 50\%, and White children at 51\%. We note that these are simple descriptive statistics that do not weigh in on the causal effect of race and ethnicity on persistence.

\begin{table}[t]
    \centering
    \caption{Summary Statistics for Placement Files}
    \begin{tabular}{lll}
    \toprule
    Name & Mean/Percent & Persistence \\ \hline \hline
    \midrule
    \textbf{All Children} & & 0.51 \\ \hline
    \textbf{Child Characteristics} & & \\
    Age & 9.92 & \\
    Female & 0.47 & 0.51 \\
    Male & 0.53 & 0.51 \\
    Black & 0.32 & 0.5 \\
    White & 0.63 & 0.51 \\
    Hispanic & 0.15 & 0.5 \\
    Native American & 0.04 & 0.49 \\
    Has Disability & 0.25 & 0.45 \\ \hline
    \textbf{Home Characteristics} & & \\
    Married Couple & 0.41 & 0.51 \\
    Unmarried Couple & 0.02 & 0.51 \\
    Single Parent & 0.21 & 0.5 \\
    Institutional & 0.21 & 0.5 \\
    Group & 0.14 & 0.53 \\ \hline
    N Placements & 2335517 & 0 \\
    \bottomrule
    \end{tabular}
    \label{table:placement-summary-stats}
\end{table}

\textit{Generating Profiles}---In total, we generate five-hundred child profiles and five-hundred home profiles. A child (home) profile is a vector of characteristics $x_c$ ($x_h$). The $i$th characteristic of a child (home) is coordinate $x_{c,i}$ ($x_{h,i}$) of the vector. This coordinate can encode quantitative or qualitative information about the child (home). We base the characteristics on the observable features available in our data, and we randomly sample characteristic vectors from the empirical distribution. Importantly, for each market $m$, we randomly select one county and sample only children and homes from the county to increase our simulations' robustness. Since the children and homes are from the same county, the simulated matches are plausible. This precludes us from simulating outcomes in very small counties\footnote{Eight or less children in one year.} because AFCARS censors them, but results using nation-wide sampling are qualitatively and quantitatively similar. We assume that every simulated home has one available bed for a foster child for simplicity's sake. For each child and home, we write a neutral descriptive profile given the characteristics of the child (home). Therefore, when submitting preferences, every subject has the same information as used in the classification predictor.

\textit{Random Forest Predictor}---Random forests are a machine learning method for classification and regression analysis. They are collections of decision trees that partition data based on features (covariates) to minimize a loss function in predicting some outcome. Random forests use bootstrap aggregating to draw a random subsample from the data and train a decision tree on the bootstrap sample using a random selection of features. The forest predicts a classification outcome through majority vote of the ensemble of decision trees. Random forests are known to be highly accurate and minimize variance compared to many competing machine learning methods, inspiring our use of this method to predict persistence. We refer the reader to \qcite{breiman-2001} for details on training a random forest and utilizing it for prediction.

We train the random forest to predict persistence on a random subsample using 75\% of the placements. We tune our parameters for the random forest to use one-hundred trees, maximum depth and splitting, a Gini loss measure, and a maximum of square root of the total number of features. We use all features in the Placement Files for our training and display their importances in Figure \ref{fig:feature-importances}. After training the random forest, we assess its validity in the remaining 25\% test sample. The accuracy in our test sample is a remarkable 77\%. We report the confusion matrix in Figure \ref{fig:confusion-matrix}. The false negative rate is approximately 15\%, and the false positive rate is approximately 7.8\%.

\begin{figure}[t]
    \caption{Feature Importances}
    \includegraphics[width=\linewidth]{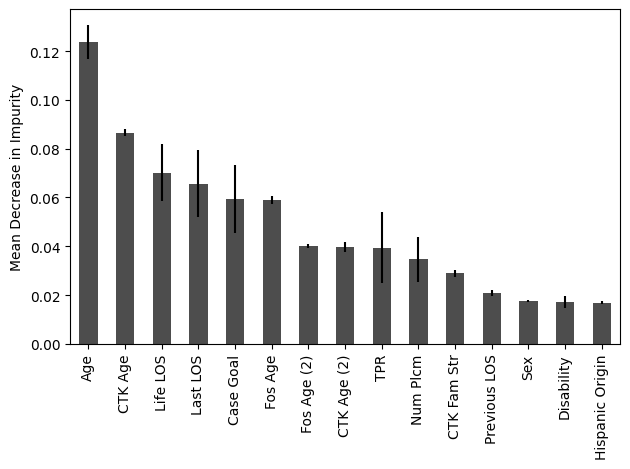}
    \floatfoot{Note: Life LOS is the number of days a child spent in foster care; for Last LOS this is the LOS since the child's most recent removal from her home of origin. Previous LOS is the length of stay in the child's previous foster care spell. CTK Age, CTK Age (2), and Fos Age refer to the first caretaker's, second caretaker's, and first foster caretaker's year of birth since the Placement Files are cross-sectional and do not include the placement year as a feature. Setting refers to the placement type (family home, institutional home, or group home). Poor housing refers to whether or not inadequate housing was cited as a reason for the child's removal from the biological parents.}
    \label{fig:feature-importances}
\end{figure}

\begin{figure}[t]
    \caption{Confusion Matrix for Random Forest Model}
    \includegraphics[width=\linewidth]{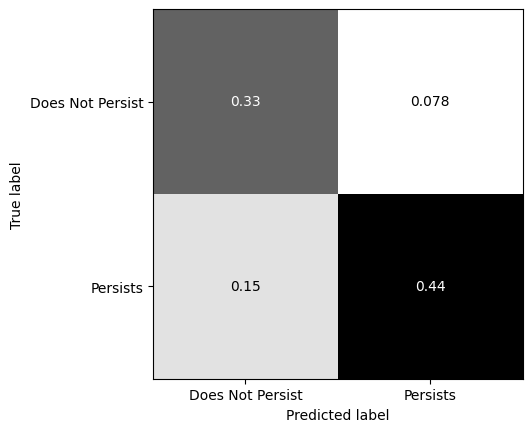}
    \floatfoot{Note: Persistence averages sum to different values than Table \ref{table:placement-summary-stats} because this matrix is subset to the test data.}
    \label{fig:confusion-matrix}
\end{figure}

Random forests are adept at identifying interactions in the data without the inclusion of interaction terms; this implies that, at the placement level, our predictor can identify relevant match-specific features that impact persistence. This point is crucial because gains from improved matching procedures are more relevant when persistence is heterogeneous across children. To measure the degree of horizontal differentiation in persistence, we simulate one hundred small markets in the same manner that we generate profiles, denoting a typical market as $s$. For a fixed $c,h$, we calculate the shared persistence rate as $p_s(c,h) = 1/(|C_s| - 1) * \sum_{c' \in C_s - c} \Hat{Y}(c, h)$. Intuitively, $p_s(c,h)$ is the percentage of children in the market other than $c$ that would persist if matched to $h$. We define $c$'s persistence set as $H^p_s(c) = \{ h \in H_s : \Hat{Y}(c,h) = 1 \}$. $c$'s average shared persistence rate is $p_s(c) = 1 / |H^p_s(c)| * \sum_{h \in H^p_s(c)} p_s(c,h)$. Under perfectly vertically (horizontally) differentiated persistence, $p_s(c) = 1$ ($0$) for all $c$. The persistence heterogeneity index is $\Bar{p}_s = 1/|C_s| * \sum_{c \in C_s} p_s(c)$, and the same interpretation for this index holds. Intermediate values of the index measure the degree of horizontal differentiation in the persistence in the market. The results for our persistence heterogeneity simulation are in Figure \ref{fig:persistence-heterogeneity-5}. The average over all simulated markets is $0.335$.

\begin{figure}[t]
    \caption{Persistence Heterogeneity Simulations}
    \includegraphics[width=\linewidth]{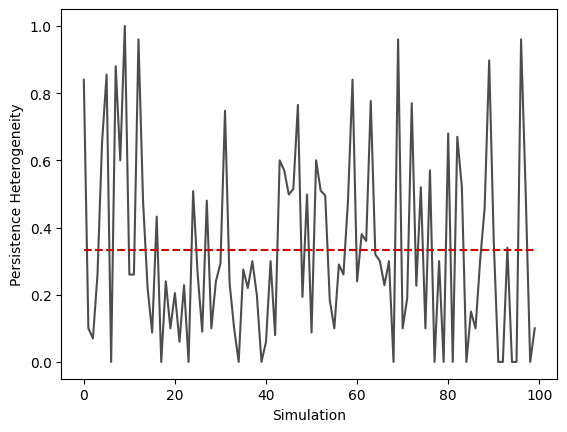}
    \floatfoot{Note: Small market size mechanically causes high variation in the score over simulations.}
    \label{fig:persistence-heterogeneity-5}
\end{figure}

Figures \ref{fig:feature-importances} and \ref{fig:persistence-heterogeneity-5} suggest that match-specific effects impact persistence. First, while the feature importance metric cannot interpret the impact of interaction terms, features relevant to foster homes (foster caretaker age) and the child (age, family structure, etc) appear in the top fifteen. These two facts imply \textit{potential} for interactions exists. Next, the simulation results affirm this potential. In the average small market, the average child shares less than a third of her positive persistence outcomes. Therefore, for the average child, there are many homes that she would persist at that other children would not persist at. To our knowledge, the literature on child welfare---in and outside of economics---has not assessed the significance of ex-ante known match-specific features on persistence. Our preliminary checks here show that these match-specific features can, in practice, be a predictive determinant of persistence.

In our counterfactuals, we use persistence as the evaluation for mechanisms. However, our random forest model is a classifier predicting a binary outcome. Therefore, we have to break indifferences to create strict evaluations. The underlying architecture of random forests consists of many decision trees that vote on the classification. If the majority vote is one, then the classification is one and vice-versa for zero. We break indifferences by using the random forest's majority votes for a given classification. This approach also has the advantage that more confident predictions are ranked higher. We provide the empirical prediction error distribution by random forest vote decile in Figure \ref{fig:error_distribution}. The error rate peaks around the $0.5$ cutoff and falls toward the tails, implying that errors might have limited impact on unanimous mechanisms which emphasize higher ranked options.

\begin{figure}[t]
    \caption{Error Distribution by Vote Decile}
    \includegraphics[width=\linewidth]{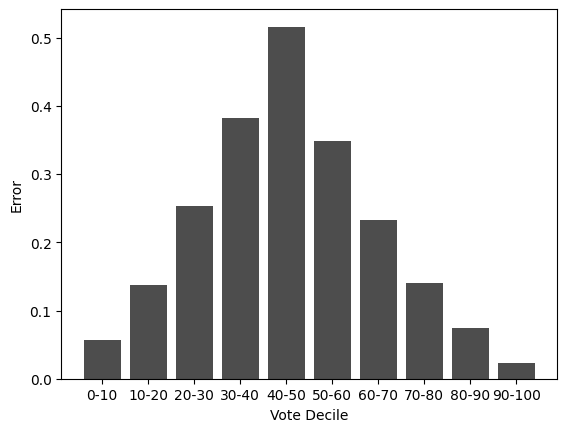}
    \floatfoot{Note: The error rate is the percentage of incorrect predictions within the decile.}
    \label{fig:error_distribution}
\end{figure}

Furthermore, in robustness specifications in our results, we perturb the persistence $Y(c,h)$ for each $c$ and $h$ to induce misclassification. Using this method, we can ascertain the impact of prediction error on child welfare when the matchmaker uses predicted persistence as the evaluation. Our results are robust to the observed empirical distribution of misclassification in the out-of-sample testing, further strengthening the case for utilizing our mechanisms.

\subsection*{4.3 \hspace{5pt} Experiment Procedure}

Before conducting the experiment, we performed simulations and power analyses to assess the optimal sample size. The simulations are in Appendix C. We determine that a minimum of three-hundred subjects would be ideal to detect effects in our experiment. We recruit caseworkers through child welfare agencies in counties across the United States. We reached out to these counties to disseminate an email invitation to our experiment sent to caseworker employees. Subjects that read and responded to our email were able to immediately participate in the online experiment.

The experiment was available from DATE to DATE. We allowed subjects to asynchronously complete the experiment. The experiment begins with an instructions interface that explains relevant terminology, the economic environment, and rules. While or after reading the instructions, the subject fills out (incentivized) control questions that verifies her understanding. Once complete, we ask subjects to answer a basic demographics survey that includes an additional question soliciting the subject's experience as a foster care caseworker.

After finishing the instructions and survey, the subject proceeds to the experimental interface to complete the five rounds of the experiment. Subjects were allotted ten minutes per round to read their child profile, home profiles, and submit preferences. The average time to complete a round was TIME. Once a subject completed her rounds, she receives her payoff and exits the experiment immediately. We pay subjects a baseline participation fee of ten dollars USD plus her gains or loses from incentive pay. The conversion rate for one experimental unit was ten-cents USD. Surveys and experiment interfaces are available in the Online Appendix.


\section*{5 \hspace{5pt} Conclusion}

Before concluding, we highlight a number of salient future directions for our simulations. As mentioned in the Introduction, our tentative plan involves running a lab-in-the-field experiment with real caseworkers to evaluate their preferences and decision-making in a simulated market where children's predicted persistence is determined through a random forest predictor that we train on a national dataset with several million observations. It achieves an 77\% out-of-sample prediction accuracy. Combining these elements, we can estimate counterfactual matchings and persistence in the simulated market to provide plausible indications on the welfare improvements from using unanimous matching mechanisms in the child welfare context. Furthermore, we can use surveyed information about participants to assess how experience affects within-caseworker average persistence.

Our work in this paper establishes matching mechanisms that are appropriate for foster care while also improving placements beyond the status quo (Serial Dictatorship) as measured by reducing the number of placements that a child experiences while in foster care. In our preliminary simulations, placement stability--or persistence---can increase by up to 28.3 pp in small markets with five to ten children matched every day when the algorithm is accurate and at worst 6.5 pp when the algorithm is poor in cases when the caseworkers are accurate. We note that in our work on the experiment so far, when assessing persistence with our random forest predictor and simulating caseworker preferences by assuming that they racially match, we find that average persistence increases by 11.2 pp. We use a new properties, unanimity and unimprovability, to characterize mechanisms that fairly assign children to homes where caseworkers and the algorithm do not concur that there is a better home for the child. Our flagship mechanism, ASDI, constructs unimprovable matchings and improves average persistence.

Nevertheless, no mechanisms in this setting are meaningfully strategy-proof in any way unless the number of homes is less than or equal to the number of children. Generally, this is true in foster care. We generalize our model to allow any preference aggregation rule to check the commonality of strategy-proofness for rules other than unanimity. We find that no mechanism and aggregation rule satisfy a weakened version of strategy-proofness---group robustness---that only allows manipulations that improve all children. This demonstrates that unanimity is not unique in its failure to mitigate strategic incentives. We highlight that further research could expand on our generalized model and explore whether or not other aggregation rules feature any other positive, useful properties.

As a final note, while our application focuses on child welfare, our framework can be applied to many other one-sided matching markets where there are multiple preference orderings for one side of the market. Examining other applications is another area with great potential for future research.

\printbibliography

\appendix
\counterwithin{table}{section}
\counterwithin{figure}{section}
\counterwithin{theorem}{section}
\counterwithin{proposition}{section}
\counterwithin{lemma}{section}

\section{Proofs}

Proofs that did not appear in the main text are detailed here.

\noindent\textbf{Proposition \ref{proposition:unanimous-reduction}}. A matching $\mu$ is unanimous on $L$ if and only if it is efficient on $L' = (M, >_I, >)$.

\begin{proof}
    We first prove $(\implies)$, that is, if $\mu$ is unanimous on $L$ then it is efficient on $L'$. Suppose, for a contradiction, that this is not true. This would imply that $\mu$ is unanimous on $L$, but $\mu$ is not efficient on $L'$. If $\mu$ is not efficient on $L'$, then there exists some sequence $(c_1, c_2, ..., c_n)$ such that for all $i \neq n$, (i) $\mu(c_{i+1}) \geq_{I(c_i)} \mu(c_i)$, (ii) either $\mu(c_1) \geq_{I(c_n)} \mu(c_n)$ or $h \geq_{I(c_n)} \mu(c_n)$ for some $h$ with $\mu(h) = \varnothing$, and (iii) this must hold strictly for at least one $i$. Consider the matching $\mu'$ such that the (possibly cyclic) trades implied by the sequence above occur. For any $c$ such that $\mu(c) = \mu'(c)$, we have that $c \in I^*(\mu) \iff c \in I^*(\mu')$ and $c \in \kappa(\mu) \iff c \in \kappa(\mu')$. For any $c'$ such that $\mu(c') \neq \mu'(c')$, we have that $\mu'(c') \geq_{I(c')} \mu(c')$, so:
    \[c' \in I^*(\mu) \implies \mu(c') \in H^*(c') \implies \mu'(c') \in H^*(c') \implies c' \in I^*(\mu')\]
    the following also holds:
    \[c' \in \kappa(\mu) \implies \mu(c') \neq \varnothing \implies \mu'(c') \neq \varnothing \implies c' \in \kappa(\mu')\]
    For at least one $c''$: $\mu'(c'') >_{I(c'')} \mu(c'')$. We have that either $\mu(c'') \notin I^*(\mu)$ or $\mu(c'') = \varnothing$. We also have that if $\mu(c'') \notin I^*(\mu)$, then $\mu(c'') \in I^*(\mu')$. If $\mu(c'') = \varnothing$, then $\mu'(c'') \neq \varnothing$. These facts imply that either $I^*(\mu) \subset I^*(\mu')$ or $\kappa(\mu) \subset \kappa(\mu')$, i.e., $\mu$ is not unanimous, a contradiction. 
    
    Next, we prove $(\impliedby)$, that is, if $\mu$ is efficient on $L'$ then it is unanimous on $L$. Again, suppose, for a contradiction, that $\mu$ is efficient on $L'$ but that it is not unanimous on $L$.  Since $\mu$ is not unanimous, there exists some $\mu'$ such that either (a) $I^*(\mu) \subset I^*(\mu')$ and $\kappa(\mu) \subseteq \kappa(\mu')$ or (b) $I^*(\mu) \subseteq I^*(\mu')$ and $\kappa(\mu) \subset \kappa(\mu')$. Without loss of generality, we denote the set $S$ and $S'$ as $I^*(\mu)$ and $I^*(\mu')$ if (a) holds, alternatively, $\kappa(\mu)$ and $\kappa(\mu')$ if (b) holds. First, note that if $\mu'(c) = h$ for some $c \in S'$ such that $c \notin S$ and $h$ with $\mu(h) = \varnothing$, then we can construct a third matching $\mu''$ where $\mu''(c) = h$ and $\mu'' = \mu$ otherwise. The above relations give us that $\mu''(c) >_{I(c)} \mu(c)$ and $\mu''(c') = \mu(c')$ for all $c' \neq c$, so $\mu''$ Pareto dominates $\mu$ on $L'$, contradicting efficiency. Therefore, for any $c$ satisfying the stated conditions, it must be that $\mu'(c) = h$ for some $h$ such that $\mu(h) \neq \varnothing$. We can construct a sequence $(c_1, c_2, ..., c_n)$ such that $c_1 = c$ and $c_i = \mu(\mu'(c_{i-1}))$ for $i \neq 1$. By unanimity of $\mu'$, $\mu(c_i) = \mu'(c_{i-1}) \geq_{I(c_{i-1})} \mu(c_{i-1})$ for $i \neq 1$ and $\mu(c_2) = \mu'(c_1) >_{I(c_1)} \mu(c_1)$. Yet, if $\mu$ is efficient on $L'$, then no such sequence described can be constructed, a contradiction. This proves the Proposition.
\end{proof}

\noindent\textbf{Proposition \ref{proposition:uttc-efficient-unanimous}}. Unanimous Top Trading Cycles is constrained-efficient and unanimous.

\begin{proof}
    Unanimity follows because for any $c \in I^*(\mu^1)$, $c$ can only participate in a cycle or chain such that $c$ would be matched to some $h \in H^*(c)$. A similar conclusion holds for every $c \in \kappa(\mu^1)$. Unanimity of $\mu^1$ directly implies unanimity for $\phi_{UTTC}$.

    Next, we demonstrate that UTTC is constrained-efficient. Suppose, for a contradiction, that it is not. There must exist some $\mu'$ such that $\phi_{UTTC} \vartriangleright_{PD} \mu'$. Since $\phi_{UTTC}$ is unanimous, $\mu'$ cannot (a) weakly expands the set of unanimous children and matched children and (b) strictly so for one of the two sets. Therefore, $I^*(\phi_{UTTC}) = I^*(\mu')$ and $\kappa(\phi_{UTTC}) = \kappa(\mu')$. $\phi_{UTTC} \vartriangleright_{PD} \mu'$ implies that we can identify some Pareto improving cycle or chain $(c_1, c_2, ..., c_n)$ such that every $c_j$ in the cycle satisfying $c_j \in I^*(\phi_{UTTC})$ is reallocated to a unanimous home. 
    
    Consider the earliest child $c_i$ in the cycle to be removed from the UTTC algorithm. If $c_i \in I^*(\phi_{UTTC})$, then $\phi_{UTTC}(c_{i+1}) \in H^*(c_i)$. Since $c_i$ was removed, this implies that $c_i$ was not pointing at $c_{i+1}$. Further, because $\phi_{UTTC}(c_{i+1}) \in H^*(c_i)$ we know that it must be the case that $c_{i+1}$ was removed in an earlier round. This contradicts our assumption that $c_i$ was the earliest to be removed. If $c_i \notin I^*(\phi_{UTTC})$, the same contradiction is immediate because $c_i$'s pointing is only restricted to unanimously acceptable homes. Therefore, no such cycle or chain can exist, and $\phi_{UTTC}$ cannot be Pareto dominated by a matching that weakly expands the set of unanimous and matched children. This proves the Proposition.
\end{proof}

\noindent\textbf{Theorem \ref{theorem:idsdi}.} Adaptive Serial Dictatorship with Indifferences is unimprovable.

\begin{proof}
    We prove the Theorem by contradiction. First, note that if $\phi_{ASDI} \equiv \mu$ is improvable, then we can find a cycle $Y = (c_1, c_2, ..., c_n)$ that satisfies $I(c_i, \mu(c_{i+1})) \subseteq I(c_i, \mu(c_i))$ for all $c_i \in Y$, and $\mu(c_{i+1}) \in I(c_i, \mu(c_i))$ for some $c_i \in Y$, where we define $c_{n+1} = c_1$. Alternatively, the above holds for all $c_i \neq c_n$ and for some $c_i$ (respectively), but, for some $h$ with $\mu(h) = \varnothing$, $I(c_n, h) \subseteq I(c_n, \mu(c_n))$.

    Let $w$ be the matching implied by the reallocation $Y$ and all children not in $Y$ keep the same match. Consider the earliest dictator $c$ in a round $n$ such that $w(c) \in I(c, \mu(c))$. If $w \in U_c$, then $w(c) \in H^n$ implying either $w(c) >_{I(c)}^n \mu(c)$ or some $h \in H^n$ exists such that $h >_{I(c)}^n w(c) \implies h >_{I(c)}^n \mu(c)$. In either case, $\mu \notin U_c$, a contradiction. Therefore, $w \notin U_c$.

    Hence, some dictator $c'$ in an earlier round $k < n$ did not choose $w$, i.e.,
    \[\mu(c') >_{I(c')}^k w(c')\]
    If $w(c') \neq \varnothing$, the above implies that there exists some home $i \in H^k$ such that $i \in I(c', w(c'))$. $\mu(c') >^k_{I(c')} w(c')$ and $w(c') \neq \varnothing$ imply that $\mu(c') \sim^k_{I(c')} i$ since there are only two indifference classes. If $w(c') = \varnothing$, we trivially have that $\mu(c') \in I(c', w(c'))$ and $\mu(c') \notin I(c', \mu(c'))$. So, $i \notin I(c', \mu(c'))$. Therefore, no such cycle $Y$ can exist.
\end{proof}

\noindent\textbf{Theorem \ref{theorem:obvious-manipulability}.} The following statements are true:
    \begin{enumerate}
        \item[i.] Any unanimous matching mechanism is obviously manipulable.
        \item[ii.] Serial Dictatorship with Indifference is non-obviously manipulable if $|H| \leq |C|$.
    \end{enumerate}

\begin{proof}
    We prove (i) that any unanimous matching mechanism is obviously manipulable using an example. Consider the following market with $C = \{ a, b \}$ and $H = \{ 1, 2, 3 \}$. The preferences, evaluations, and induced unanimity orders for $b$ are as follows, where omitted homes are not unanimously acceptable and brackets indicate indifference:
    \begin{align*}
        &b\\
        \succ_b: 1, 2&, 4, 3\\
        >_b: 3, 1&, 2, 4\\
        >_{I(b)}: [1, 3&], [2, 4]
    \end{align*}
    Let $\phi$ be any unanimous matching mechanism. Every matching that is consistent with it must be unanimous. Under the truthful report for $b$, in any unanimous matching $\mu$, either $\mu(b) = 1$ or $\mu(b) = 3$. This implies that the worst-case scenario in any unanimous matching is that $b$ is matched to $3$. Under the report $\succ'_a$ where $a$ reports only $1$ as acceptable, $H^*(a) = \{1\}$. Therefore:
    \[CU^-[\succ_b] = \{ (1, 3) \}\]
    This implies that the worst-case under truth-telling is no better than $b$ matching to $3$. Consider a profile $\succ'_b$ where $b$ reports the order $\succ'_b : 4,2,1$. Under this report, every home is unanimous except for $3$. Thus, in any unanimous matching $\mu'$, either $\mu'(b) = 1$ or $\mu'(b) = 2$, or $\mu'(b) = 4$. The worst-case under $\succ_b'$ must be better than $3$, so this manipulation is a worst-case manipulation.
    
    Last, we prove (ii) in a general environment for any child $c$. It is immediate that the best case for $c$ under truthful reporting is going first in the order and matching to her top choice which is always unanimous. Under SDI, we can stipulate preferences for all other $c' \neq c$ such that $c$'s top choice is unacceptable. Therefore, $c$ will also be matched to her top choice under SDI in the best-case scenario.

    Suppose that $|H| < |C|$. The worst-case scenario under any report is that $c$ goes last and that every preceding child has a single top choice such that the set $H$ is covered by the children's top choices. In this case, $c$ is unmatched under either mechanism. If $|H| = |C|$, under any report, the same argument applies to show that $c$ is always matched to her least favorite home if $c$'s report is such that all homes are unanimously acceptable and is unmatched if at least one home is not unanimously acceptable.
\end{proof}

\noindent\textbf{Theorem \ref{theorem:gs-aggregation}.} When $|H| \geq |C| + 2$, for any $\tau$ satisfying the weak Pareto principle, an efficient, strategy-proof, and IIA mechanism $\phi$ is group robust if and only if $\tau$ is a dictatorial serial choice rule.

\begin{proof}
    We prove that if $\phi$ is efficient, strategy-proof, and IIA, then it is only group robust for a $\tau$ satisfying WPP if $\tau$ is a dictatorial serial choice rule. First, we show that the rule is a serial choice rule and prove that it is dictatorial in a later step. We will prove the Theorem for the simplifying case where $|C| = 1$ and $|H| \geq 3$. Then, we extend it to any case where $|C|$ is finite and positive and $|H| \geq |C| + 2$.

    Let $C = \{c\}$ and $1,2 \in H$ be two distinct homes. Consider the rankings:

    \[\succ_c : 1, 2, j\]
    and
    \[>_c : 2, 1, j\]
    where for any other $j \in H$, $2 \succ_c j$ and $1 >_c j$.

    \textbf{Claim 1:} if $\tau$ follows WPP, then $\phi(c) \in \{1, 2\}$ for any efficient (with respect to $L_\tau$) and group robustness $\phi$. WPP implies that $1 T_c j$ and $2 T_c j$. Since either $1$ or $2$ will always be unmatched, the claim immediately follows by efficiency.

    Suppose that $\phi(c) = 2$. We have that $2T_c1$ by efficiency. 
    
    \textbf{Claim 2:} Given that $\phi(c) = 2$, for any arbitrary preference $\succ''_c$ and any evaluation $>''_c$ where $2$ is top-ranked with corresponding aggregation $T''_c$, $2 T''_c i$ for any $i \neq 2$
    
    Consider a report for $c$:
    \[\succ^J_c : 1, j, 2\]
    i.e., $1 \succ^J_c j$ and $j \succ^J_c 2$ for all $j \neq 1,2$. Denote $T^J_c \equiv \tau(\succ^J_c, >_c)$. By group robustness, we still have that $\phi[M, \tau(\succ^J_c, >_c)](c) \neq 1$. If $\phi[M, \tau(\succ^J_c, >_c)](c) = j$, then we must have that $j T^J_c 1$. However, by WPP, we have that $1 T^J_c j$, a contradiction. Therefore, $\phi[M, \tau(\succ^J_c, >_c)](c) = 2$.

    Now, we will consider a series of preferences and evaluations to show the Claim. We denote these as sequences to shorten exposition. Let the binary relation representing the aggregation in each case be $T''_c$.

    (1) Preferences are $(1, i)$ where $1$ is preferred to any $i \neq 1$ and each $i$ is positioned arbitrarily. Evaluations are $(2, 1, i)$. By the same logic as above, the match must still be $2$.

    (2) Preferences are $(1, i)$ and evaluations are $(2, i)$. By group robustness, the match must still be $2$, otherwise the matchmaker could deviate to the report in (1).

    (3) Preferences are some arbitrary $\succ''_c$ and evaluations are $(2,i)$. Suppose that the match is not $2$. In the case where $c$'s preferences are as in $\succ_c^J$, $c$ has an incentive to deviate to $\succ''_c$, contradicting group robustness. The match must still be $2$. Therefore, for any arbitrary preference $\succ''_c,>''_c$ where $2$ is top-ranked we have that $\phi[\succ''_c, >''_c](c) = 2$. This implies that $2 T''_c i$ for any $i \neq 2$.

    \textbf{Claim 3:} Given that $\phi(c) = 1$, for any preference $\succ''_c$ where $1$ is top-ranked and any arbitrary evaluation $>''_c$ with corresponding aggregation $T''_c$, $1 T''_c i$ for any $i \neq 2$. The argument proceeds as in Claim 2 symmetrically.

    We can apply these arguments to any pair of homes $i,k$, and, by Claim 1, either $\phi(c) = i$ or $\phi(c) = k$. This will partition $H$ into sets $H_2$ (homes such that Claim 2 will apply) and $H_3$ (homes such that Claim 3 will apply). Suppose that $|H_2| > 0$ and $|H_3| > 0$. For any generic elements $i \in H_2$ and $k \in H_3$, we can take arbitrary evaluations and preferences for $c$ such that $i$ and $k$ are at the top, respectively. By each claim, we must have that $\phi(c) = i$ and $\phi(c) = k$, which is impossible. Therefore, either $H_2 = H$ or $H_3 = H$ but not both. Without loss of generality, suppose that $H_3 = H$. Then, $i T_c j$ for any $j \neq i$ when $i$ is top-ranked under preferences.

    For $|C| = n$, we use induction. First, notice that one can extend the above argument to any $|C|$ by inserting an additional child and home such that the child's only acceptable home under both rankings is the additional home. However, we only proved that $\tau$ is dictatorial for the top choice. We extend this to $\#\tau^D \geq n$. 
 
    \textbf{Induction Statement}: for a market with $|C| \geq n$, an efficient and IIA $\phi$ can only be group robust if $\tau$ is equivalent to a serial choice rule where $\#\tau^D \geq n$. We proved the base case $n = 1$ above. We show it for $n+1$ assuming that it holds for all integers $n' \leq n$.

    Consider a hypothetical market where $|C| = n+1$. For all $c' \in C$, we will specify the following rankings:
    \[\succ_{c'} : 1, 2, ..., n, x, y, j\]
    and
    \[>_{c'} : 1, 2, ..., n, y, x, j\]
    where exactly $n$ homes precede $x$ under the preferences and $n$ homes precede $y$ under the evaluations, and $j$ is defined as in the case with $|C| = 1$.

    By the WPP, we have that $i T_{c'} i+1$ for all integers $i \in [1, n]$. By efficiency, $\phi(i) \neq \varnothing$ . We also have that there exists some $c \in C$ such that $\phi(c) \neq i'$ for any such $i'$ by the pigeonhole principle. By the same logic as Claim 1 above, $\phi(c) \in \{x, y\}$.
    
    \textbf{Claim 4:} Given that $\phi(c) = y$, for any arbitrary preferences $\succ''_c$ and any evaluation $>''_c$ where $y$ is ranked before $x,j$, $y T''_c x$ and $y T''_c j$.

    Consider a report for $c$:
    \[\succ^J_c : 1, 2, ..., n, x, j, y\]
    If $\phi[M, \tau(\succ^J_c, >_c)](c) \succ_c x$ then $\phi[M, \tau(\succ^J_c, >_c)](c) = n'$ for some $n' \leq n$. However, this contradicts IIA because $n' T_c^J i$ and $n' T_c i$ for any $i > n'$, but switching back to the report $\succ_c$ changes $c$'s match. So we have that $x \succsim \phi[M, \tau(\succ^J_c, >_c)](c)$. If $\phi[M, \tau(\succ^J_c, >_c)](c) = j$, then we must have that $j T^J_c x$. However, by the WPP, we have that $x T^J_c j$, a contradiction. This implies that $\phi[M, \tau(\succ^J_c, >_c)](c) \in \{ x, y \}$. 
    
    Suppose that $\phi[M, \tau(\succ^J_c, >_c)](c) = x$. By group-robustness, it must be that some $c' \neq c$ has $\phi(c') \succ_{c'} \phi[M, \tau(\succ^J_c, >_c)](c')$. However, again, this violates IIA because switching back to $\succ_c$ changes the match for $c'$ when $c$ moves to the unmatched home $y$. Therefore, $\phi[M, \tau(\succ^J_c, >_c)](c)  = y$.
    
    Next, we follow the same proof strategy as in Claim 2. Consider the following series of preferences and evaluations for $c$:

    (1) Preferences are $(1, 2, ..., n, x, i)$ where $x$ is preferred to any $i > n$ for arbitrarily positioned $i$. Evaluations are $(1, 2, ..., n, y, x, i)$. By the same logic as above, the match must still be $y$.

    (2) Preferences are $(1, 2, ..., n, x, i)$ and evaluations are $(1, 2, ..., n, y, i)$. We can argue as above that the match cannot be any $n' \leq n$. The match must be $y$ or some $i$. Suppose that it is some $i$. By IIA, matches for all other children remain the same since $i$ must be unmatched in the previous case (1) and $y$ must be unmatched in this case (2). The matchmaker could deviate to the report in (1) and improve to $y$ without changing the match for any other child, violating group robustness. Therefore, the match must be $y$.

    (3) Preferences are some arbitrary $\succ''_c$ and evaluations are $>''_c \equiv (1, 2, ..., n, y, i)$. Suppose that the match is not $y$. First, note that even if $T''_c = \succ_c$, we cannot have that $c$'s match is some $n' \leq n$ because then, by IIA, $c$'s match must also be $n'$ under $T_c$. Strategy-proofness implies that no report can improve $c$'s match beyond truthful reporting, so the match is not any such $n'$. If the match is not $y$ then it is some $j > n$ (potentially including $x$). Furthermore, by efficiency we had that that $\phi[M, \tau(\succ^J_c,>''_c)](j) = \varnothing$, and we must have that the first $n$ homes are still matched. Efficiency also implies that $\phi[M, \tau(\succ''_c, >''_c)](y) = \varnothing$. By IIA, all other $c' \neq c$ have the same match under $T''_c$. Suppose that $c$'s true preferences were $\succ^J_c$. $c$ would have an incentive to misreport $\succ''_c$ and would not harm any other children, violating group robustness. Therefore, we must have that the match is still $y$.

    (4) Preferences are some arbitrary $\succ''_c$ and evaluations are $(n', y, i)$ for some arbitrary re-positioning of the first $n$ elements. Suppose that the match is not $y$. By the same argument as in (3), it cannot be any $n' \leq n$. Additionally, by the same argument as in (3), if it is some $j > n$ for $j \neq y$, then the matchmaker has an incentive to misreport $(1, 2, ..., n, y, i)$ without harming any other child. Therefore, the match must still be $y$. Hence, for any arbitrary $\succ_c''$ and any $>''_c$ where $y$ is ranked above any $i > n$, we have that $y T''_c i$.

    Claim 4 implies that $\tau$ is dictatorial if $\phi(c) = y$. To see this, observe that it must hold for the evaluation $>''_c = >_c$ above and any preferences:
    \[\succ''_c : 1, 2, ..., x, n', ..., n, y\]
    where $n' \leq n$. By the induction statement, for $|C| = n+1$, at least the first $n$ homes in the ranking $\tau$ are determined by a serial choice rule. $y$ cannot be chosen since it is not in the first $n$ rankings under either ranking. However, we still have that $y T''_c x$, implying that $x$ also cannot be chosen. Since this holds for an arbitrary $n' \leq n$, the child can never be the dictator. Furthermore, since the order of dictators is fixed under a SCR, she can never be the dictator for any preferences.

    \textbf{Claim 5:} Given that $\phi(c) = x$, for any preference $\succ''_c$ where $x$ is ranked before $y,j$ and any arbitrary evaluation $>''_c$, $x T''_c y$ and $1 T''_c j$. The argument proceeds as in Claim 4 symmetrically. We also have that if $\phi(c) = x$, then the matchmaker can never be the dictator.

    Fix the first $n$ homes as we initially did for preferences and evaluations. We can apply these arguments to any pair of homes $i,k > n$, and, by Claim 4, either $\phi(c) = i$ or $\phi(c) = k$. This will partition all such homes into sets $H_4$ (homes such that Claim 4 will apply) and $H_5$ (homes such that Claim 5 will apply). Suppose that $|H_4| > 0$ and $|H_5| > 0$. For any generic elements $i \in H_4$ and $k \in H_5$, we can define:
    \[\succ^4_c : 1, 2, ..., n, i, k, j\]
    and
    \[>^5_c : 1, 2, ..., n, k, i, j\]
    By each claim, we must have that $\phi(c) = i$ and $\phi(c) = j$, which is impossible. Therefore, either $|H_4| > 0$ or $|H_5| > 0$ but not both.

    Without loss of generality, suppose that $H_5 = H$. Since $|H| \geq n + 3$, we can take some $i,k \in H_5$ where $i$ is ranked in position $n+1$ in $c$'s preferences $\succ_c''$ so that Claim 5 applies for any arbitrary $>_c''$. This will imply that $i T_c m$ for any $m \neq i,k$ such that $m > n$. However, if we take some $i,j \in H_5$, then we also get that $i T_c m'$ for any $m' \neq i,j$ and $m' > n$. So $i T_c k$ and $i T_c j$. Therefore, $i T_c j'$ for any $j' \neq i$ and $j' > n$. The first $n$ homes in $\tau$ are chosen dictatorially, so $i$ must be ranked exactly in position $n+1$ as desired.

    $\tau$ is a function that only depends on preferences and evaluations, so, even when varying the problem, $\tau$ must be the same. Furthermore, the initial choice of the first $n$ homes was arbitrary, so we can repeat this entire argument for any permutation of preferences to show that $\tau$ is always dictatorial. Finally, one can insert any number of additional dummy children or homes as described above to meet the statement for any $n'' > n + 1$. This proves the induction statement.

    The other direction is immediate. If either side is the dictator, then a $\tau$-efficient and strategy-proof mechanism is $\succ$ or $>$-efficient. Therefore, there cannot be any group-beneficial manipulations for the dictatorial side. Furthermore, the non-dictatorial side cannot manipulate $\phi$ because its reports are irrelevant.
    
    This proves the Theorem.
\end{proof}

\noindent\textbf{Proposition \ref{proposition:extended-unanimity}.} The following statements are true:
\begin{enumerate}
    \item[i.] $\tau^u$ is a strict total order,
    \item[ii.] $\tau^u$ is regular, and
    \item [iii.] If $\mu$ is constrained-efficient and unanimous, then $\mu$ is efficient under $\tau^u$.
\end{enumerate}

\begin{proof}
    First, we show that $\tau^u$ is a strict total order, i.e., that it is an element of $P$. Let the binary relation of $\tau^u$ be $T^u_c$. We omit $c$ and simply consider an arbitrary $c$ for this part.

    \textit{Irreflexive:} If $T^u$ is not irreflexive, then there exists some $x \in H$ such that $x T^u x$. Clearly, it cannot be that $x$ is strictly preferred to itself on preferences nor evaluations. Furthermore, it cannot be that $x \in H^*$ and $h \notin H^*$.

    \textit{Transitive:} Suppose that for some $x,y,z$, we have that (i) $x T^u y$ and (ii) $y T^u z$. (i) implies that (a) $x$ is a unanimous improvement over $y$, (b) $x \in H^*$ and $y \notin H^*$, or (c) $x \succ y$. In case (a) or (b), $y \notin H^*$. Therefore, (ii) only if $y$ is a unanimous improvement over $z$. Hence, $z \notin H^*(c)$. In case (a), $x$ must also be a unanimous improvement over $z$ so that $x T^* z$. In case (b), immediately $x T^u z$.  In case (c), if $x \in H^-$ then $y \in H^-$ and (ii) only if $z \in H^-$ and $y \succ z$, hence $x \succ z$ by assumption that $\succ \in P$. If $x \in H^*$ and $y \in H^*$, then the same holds if $z \in H^*$. If $z \notin H^*$, then immediately $x T^u z$.

    \textit{Asymmetry:} This follows by irreflexivity and transitivity.

    \textit{Connected:} All homes are comparable by definition.

    Since only unanimously acceptable homes are acceptable under $\tau^u$, it cannot be that any unanimously acceptable home is not acceptable or that any unanimously unacceptable home is acceptable. Furthermore, suppose that $x \succ y$ and $x > y$ for any $x, y \in H$. It cannot be that $y \in H^*$. If $x \in H^*$, then $x T^u y$. If $x \notin H^*$, then still $x T^u y$ by definition and transitivity. These facts imply that the rule satisfies the WPP.

    Last, we prove that if $\mu$ is constrained-efficient and unanimous, then $\mu$ is efficient under $\tau^u$. Suppose, for a contradiction, that $\mu$ is constrained-efficient and unanimous, but $\mu$ is not $\tau^u$-efficient. There must exist some $\tau^u$-improving cycle $(c_1, c_2, ..., c_n)$ such that $\mu(c_{i+1}) R^u_{c_i} \mu(c_i)$ for all $i \neq n$ and $\mu(c_1) R^u_{c_n} \mu(c_n)$ with at least one strict relation. However, this implies that any child that is assigned to a unanimous home is reassigned to another unanimous home that is a Pareto improvement. Any child that is assigned to a non-unanimous home is reassigned to another non-unanimous home that is a Pareto improvement. Otherwise, if a child at a non-unanimous home improves to a unanimous home, this contradicts $\mu$ being unanimous. The reassignment will therefore have the same unanimous coalition and will Pareto dominate $\mu$, which contradicts $\mu$ being constrained-efficient.
\end{proof}

\begin{proposition}\label{proposition:rsa-sp}
    Random Serial Assignment is strategy proof when ignoring feasibility constraints.
\end{proposition}

\begin{proof}
    We prove this by counterexample. Denote the true preferences for $c$ as $\succ_c$ (with corresponding problem $L$). Suppose that for some manipulation $\succ'_c$ (with corresponding problem $L'$), $\phi_{RSA}[L'](c) \succ_c \phi_{RSA}[L](c)$. Then in some round $n$, under $L'$, RSA selects some $h'$ that $c$ prefers over the $h$ selected (or not selected) in the same round $n$ under $L$. However, if $h'$ is not selected under $L$, then $a(h') = 0$. Therefore, it could not be selected under $L'$. This proves the Proposition.
\end{proof}

\begin{proposition}\label{proposition:rsa-unique}
    Truth-telling is the unique, non-trivial weakly dominant strategy under Random Serial Assignment when ignoring feasibility constraints.
\end{proposition}

\begin{proof}
    In qualifying the proof with \textit{non-trivial}, we make the assumption that every $c$ has at least one acceptable home. Under any manipulation $\succ'_c$, $c$ either ranks some $h'$ higher than some $h$ such that $h \succ_c h'$ or $c$ reports some $h$ as unacceptable. In the latter case, it is clear that if $a(h) = 1$ and $a(h'') = 0$ for all $h'' \neq h$, then this manipulation is strictly worse than truth-telling. In the former case, suppose that, (a) for every $h''$ such that $h'' \succ_c h$, $a(h'') = 0$, (b) $a(h') = 1$, and (c) $a(h)$ = 1, then under the manipulation $h \succ_c \phi_{RSA}[L'](c)$, but under truth-telling $\phi_{RSA}[L](c) = h$. Since there is always some case where a misreport is strictly worse than truth-telling, truth-telling must be the only unique weakly dominant strategy under the non-triviality assumption. This proves the Proposition.
\end{proof}

\section{Empirical Performance}

In this section, we provide the technical details for our additional simulations with fully fabricated data and show that SDI and UTTC substantially increases welfare in the markets that we examine.

\subsection*{Simulation Data}

For our simulations, we generate environments that we expect to mimic child assignment to homes. We run simulations for a small market case with five to ten children, noting that in additional simulations we perform in larger markets, the effect magnitudes become larger in the same directions. The small market case is representative of a typical foster care matching market in small to medium sized counties\footnote{\qcite{cortes-2021} documents that Los Angeles county, one of the largest foster care counties in the U.S., had about forty placements per day between January and February 2011.}. The number of homes in a market is uniformly randomly distributed according to the same values\footnote{Systemic market imbalances do not affect our results.}. 

Our simulations assume that caseworkers and the matchmaker have utility representations for each home ($u_c(h)$ and $v_c(h)$, respectively) that give rise to orderings over homes. Further, we assume that the two parties will be generally aligned in pursuing some underlying social welfare $w_{c,h}$ that depends on the child-home match.

We simulate preferences and outcomes under eight different specifications. The first specification branch is \textit{accurate} versus \textit{noisy} evaluations. Concretely, we assume that $w_{c,h}$ is a uniformly distributed random variable with support on $[0,1]$, i.e, it is the true probability of persistence for a child of type $c$ matched to a home of type $h$. The evaluation utility $v_c(h) = w_{c,h} + \alpha\epsilon_{c,h}$ where $\epsilon_{c,h}$ is also a uniformly distributed random variable with support on $[0,1]$. We bound $v_c(h)$ to $0$ or $1$ if it exceeds either minimum or maximum value. $\alpha \in \{ 0, 0.25, 0.5, 0.75 \}$ measures the severity of the random noise, or prediction error, in this branch. On baseline, we assume that the preference $u_c(h) = \epsilon'_{c,h}$ where $\epsilon'_{c,h}$ is distributed in the same way as all other noise terms above.

Our second specification branch departs from the assumption on caseworker preferences above. We do not pursue analysis where caseworker preferences are more accurate than evaluations; the result, which is a mirror of our main result in this section, is that unanimous mechanisms decrease welfare if the decision-makers (caseworkers) using Serial Dictatorship are more accurate than the evaluations (algorithms) that bound the decision-makers. Instead, we will say that the caseworker preferences \textit{assist} the algorithm when it is incorrect. This specification assumes the following:

\[u_c(h) = \begin{cases}
    \epsilon'_{c,h} &\text{if (a) } w_{c,h} \geq 0.5 \text{ and } v_c(h) < 0.5 \text{ or (b) } w_{c,h} < 0.5 \text{ and } v_c(h) \geq 0.5 \\
    w_{c,h} &\text{otherwise}
\end{cases}\]

The caseworker is exactly correct when the algorithm's predicted persistence would be a false negative or a false positive, otherwise, her preference is pure noise. We use this modeling to understand the impact of human decision-makers when they have limited complementary with the algorithm, i.e., they can "correct" it when it is severely mistaken. 

Our last specification is vertical differentiation of homes. In this specification, there is a uniformly randomly distributed vertical quality for each home $q_h$ such that the utility for each caseworker is $u_c(h) = q_h + \epsilon'_{c,h}$. This allows us to capture scenarios where the caseworkers agree on the quality dimension for homes but have some taste heterogeneity. Vertical differentiation is an interesting challenge for unanimous mechanisms because it can create significant overlap in the set of unanimous homes, thereby preventing the mechanisms from achieving widespread unanimous matches. 

\begin{table}[t!]
    \centering
    \caption{Persistence, No Assistance vs. Assistance}
    \begin{tabular}{l@{\hskip 10ex}c@{\hskip 6ex}c@{\hskip 6ex}c@{\hskip 6ex}c}
    \hline
    \hline
    Specification & Base & $0.25$ & $0.5$ & $0.75$ \\ \hline
    \textit{Average persistence, baseline} \\
    SD & 50.5\% &  &  & \\
    SDI & 68.5\% &  &  & \\
    UTTC & 75.3\% &  &  & \\
    ASDI & 69.4\% &  &  & \\
    \textit{Average persistence, no assistance}\rule{0pt}{6ex} \\
    SD & & 49.3\% & 49.4\% & 49.8\% \\
    SDI & & 67.1\% & 62.9\% & 59.4\% \\
    UTTC & & 73.1\% & 67.0\% & 62.3\% \\
    ASDI & & 67.2\% & 63.2\% & 59.8\% \\
    \textit{Average persistence, assistance}\rule{0pt}{6ex} \\
    SD & & 52.5\% & 59.4\% & 59.9\% \\
    SDI & & 67.9\% & 67.9\% & 65.8\% \\
    UTTC & & 72.8\% & 70.1\% & 67.2\% \\
    ASDI & & 68.4\% & 68.3\% & 66.3\% \\
    \end{tabular}
    \label{table:persistence-noise-assistance}
\end{table}

In every specification, we run one-thousand simulations and compare the status quo, Serial Dictatorship, with SDI, UTTC, and ASDI. Importantly, we optimize for $>$-efficiency rather than $\succ$-efficiency to capture gains from the algorithm's accuracy. We focus on the persistence rate averaged across children and simulations under the two mechanisms.

\subsection*{Results}

In the baseline case, all of our theoretical predictions hold. First, Serial Dictatorship achieves a poor rate of average persistence at just 50\%, which is to be expected as caseworker's preferences are essentially a coin toss. Every unanimous mechanism improves when the algorithm is an accurate predictor. Somewhat surprisingly, even though SDI is not tuned for efficient matching, it still drastically improves average persistence by 18.5 pp. ASDI offers a slight benefit of 0.9 pp over SDI. This suggest that unanimity alone results in higher welfare even without $>$-efficient improvement. We see that all mechanisms match an equivalent proportion of children to unanimous homes. Since we assume complete preferences and evaluations, the overall number of children matched is the same in every mechanism. However, SDI, UTTC, and ASI all match a significantly higher proportion of children to unanimous homes. As theory predicts, SDI and UTTC are generally more unanimous than ASDI, and ASDI generates barely less unanimous improvements than SDI. UTTC decreases the number of unanimous improvements further, showing that constrained efficiency answers fairness concerns. UTTC is the best in average persistence with a marginal improvement of 5.6 pp over ASDI.

\textit{Noise and Assistance}---Despite concerns one might have about prediction error hampering our results, the unanimous mechanism's improvements are robust over all of our specifications, although the effect sizes are attenuated. When caseworkers are not complementary to the algorithm, we see that UTTC's welfare increase over SD is just 12.5 pp with $\alpha = 0.75$. However, considering the severity of the error---the noise term is up to seventy-five percent at maximum---this is excellent news for our mechanisms. They remain significantly superior to the status-quo even when a very, very poor algorithm is utilized as long as its predictions are centralized near the truth. In relative terms, SDI and ASDI experience a smaller decrease from their baseline performances. SDI's drop from baseline to $\alpha = 0.75$ is 8.9 pp. 

Under assistance, the same general patterns hold. The first main difference is that SD's performance gains under higher error. This effect is mechanical because of our assumption that caseworkers are accurate when the mechanism is not. The mechanism makes more false positive and false negative predictions under higher error, so caseworkers are more accurate more frequently. The second main difference is that the complementarity attenuates the loss in performance under higher error specifications. For instance, SDI's performance drops by 6.4 pp less from baseline to $\alpha = 0.75$ when assisted, and all mechanism's average persistence are strictly greater under assistance versus no assistance. This implies that unanimous mechanisms are effective in combining the information aggregated in preferences and evaluations.

\begin{table}[t!]
    \centering
    \caption{Matches, Baseline}
    \begin{tabular}{l@{\hskip 10ex}c@{\hskip 6ex}c@{\hskip 6ex}c}
    \hline
    \hline
    Specification & Percent Matched & Percent Unanimous & Num Improvements \\ \hline
    \textit{Baseline} \\
    SD & 90.0\% & 70.3\% & 7.839 \\
    SDI & 90.0\% & 86.8\% & 5.839 \\
    UTTC & 90.0\% & 86.8\% & 5.776 \\
    ASDI & 90.0\% & 85.3\% & 5.836 \\
    \end{tabular}
    \label{table:matched-children-baseline}
\end{table}

\textit{Vertical Differentiation}---Interestingly, we find that unanimous matching mechanisms perform even better when the caseworkers' preferences are correlated. We hypothesize that this result owes to the fact that, when caseworker preferences are noisy, simply correlating them will cause the noise to negatively impact persistence over a large number of simulations. In contrast, when we draw preferences randomly for each caseworker, poor draws have a smaller probability of negatively impacting average persistence.

\begin{table}[t!]
    \centering
    \caption{Persistence, Vertical Differentiation}
    \begin{tabular}{l@{\hskip 10ex}c@{\hskip 6ex}c}
    \hline
    \hline
    Specification & Base & Vertical Differentiation \\ \hline
    \textit{Average persistence, base vs. vertical} \\
    SD & 50.5\% & 50.1\% \\
    SDI & 68.5\% & 73.8\% \\
    UTTC & 75.3\% & 76.9\% \\
    ASDI & 69.4\% & 73.7\% \\
    \end{tabular}
    \label{table:persistence-vertical-diff}
\end{table}

\section{Experiment Details}

We provide additional detail on our experiment in this section.

\subsection*{Optimal Sample Size Analysis}

Calculating the effects of our mechanisms requires computing matchings for many markets under them. Caseworker preferences are necessary to do this. We fill in caseworker preferences with a plausible hypothesis---race-based matching. Anecdotally, we are aware that caseworkers in some counties use race as a criteria to attempt to match a child to a culturally relevant home. \qcite{labrenz-2022} used an IPW regression model to estimate the effects of race-based matching. When the assumptions of the model hold, they show it has a positive impact on persistence. Concretely, we assume that given some market $m$, the foster home $h$ has a parent with the same race as the child $c$ and the foster home $h'$ does not, then $h\succ_c h'$. Otherwise, $h \sim h'$. We further assume that all homes are acceptable.

We generate one-hundred small markets as in the experimental design assuming these preferences. We compare SD with $\succ$ as input to SDI, UTTC ($>$), and ASDI. We also contrast these against the "optimal" mechanism using algorithmic predictions, that is, SD with $>$ as input. Our initial findings are promising. The baseline persistence in the simulated markets is 48.6 pp. All mechanisms substantially increase persistence: SDI by 7.8 pp relative to SD, ASDI by 8.6 pp, and UTTC by approximately 8.4 pp. Both achieve a significant proportion of the $>$-efficient achievable gain, 10.4 pp.
\begin{table}[!thbp] \centering
\caption{Mechanism Performance under Race-Based Preferences}
\resizebox{\textwidth}{!}{\begin{tabular}{@{\extracolsep{5pt}}lcccc}
\\[-1.8ex]\hline
\hline \\[-1.8ex]
& \multicolumn{4}{c}{\textit{Dependent variable: Persistence}} \
\cr \cline{2-5}
\\[-1.8ex] & (1) & (2) & (3) & (4) \\
\hline \\[-1.8ex]
 SD-OPT & 0.104$^{***}$ & & & \\
& (0.038) & & & \\
 SDI & & 0.078$^{**}$ & & \\
& & (0.037) & & \\
 UTTC & & & 0.084$^{**}$ & \\
& & & (0.037) & \\
 ASDI & & & & 0.086$^{**}$ \\
& & & & (0.037) \\
\hline \\[-1.8ex]
 Observations & 200 & 200 & 200 & 200 \\
 $R^2$ & 0.037 & 0.022 & 0.025 & 0.026 \\
 Adjusted $R^2$ & 0.032 & 0.017 & 0.020 & 0.021 \\
 Residual Std. Error & 0.266 (df=198) & 0.264 (df=198) & 0.265 (df=198) & 0.264 (df=198) \\
 F Statistic & 7.638$^{***}$ (df=1; 198) & 4.375$^{**}$ (df=1; 198) & 5.037$^{**}$ (df=1; 198) & 5.288$^{**}$ (df=1; 198) \\
\hline
\hline \\[-1.8ex]
\textit{Note:} & \multicolumn{4}{r}{$^{*}$p$<$0.1; $^{**}$p$<$0.05; $^{***}$p$<$0.01} \\
\end{tabular}}
\label{table:persistence-sample-analysis}
\end{table}
Next, we run robustness checks for the case where the algorithm has prediction error. We randomly perturb the persistence $Y_c(h)$ using the empirical distribution of errors by the forest vote decile. Specifically, let $\epsilon_d$ be the empirical average error rate in prediction for decile $d$. The forest vote for a match $c$ to $h$ is $V(c,h)$ and the decile for this vote is $d_{V(c,h)}$. We flip a weighted coin that is one with probability $1 - \epsilon_{d_{V(c,h)}}$ and zero otherwise. If the coin returns zero, then we set $Y_c(h) = 1 - \hat{Y}_c(h)$. Otherwise, $Y_c(h) = \hat{Y}_c(h)$. We reconduct the analysis above with these outcomes. The results are reported in Table \ref{table:persistence-sample-analysis-robust}. We find that the positive welfare effects of our mechanisms are still significant and marginally smaller in magnitude.
\begin{table}[!thbp] \centering
\caption{Mechanism Performance under Race-Based Preferences, Robust}
\resizebox{\textwidth}{!}{\begin{tabular}{@{\extracolsep{5pt}}lcccc}
\\[-1.8ex]\hline
\hline \\[-1.8ex]
& \multicolumn{4}{c}{\textit{Dependent variable: Persistence}} \
\cr \cline{2-5}
\\[-1.8ex] & (1) & (2) & (3) & (4) \\
\hline \\[-1.8ex]
 SD-OPT & 0.086$^{***}$ & & & \\
& (0.033) & & & \\
 SDI & & 0.065$^{**}$ & & \\
& & (0.032) & & \\
 UTTC & & & 0.070$^{**}$ & \\
& & & (0.032) & \\
 ASDI & & & & 0.070$^{**}$ \\
& & & & (0.032) \\
\hline \\[-1.8ex]
 Observations & 200 & 200 & 200 & 200 \\
 $R^2$ & 0.034 & 0.020 & 0.023 & 0.023 \\
 Adjusted $R^2$ & 0.029 & 0.015 & 0.018 & 0.018 \\
 Residual Std. Error & 0.230 (df=198) & 0.229 (df=198) & 0.229 (df=198) & 0.230 (df=198) \\
 F Statistic & 6.980$^{***}$ (df=1; 198) & 4.045$^{**}$ (df=1; 198) & 4.706$^{**}$ (df=1; 198) & 4.613$^{**}$ (df=1; 198) \\
\hline
\hline \\[-1.8ex]
\textit{Note:} & \multicolumn{4}{r}{$^{*}$p$<$0.1; $^{**}$p$<$0.05; $^{***}$p$<$0.01} \\
\end{tabular}}
\label{table:persistence-sample-analysis-robust}
\end{table}
Finally, we use this sample to conduct our optimal sample size analysis. We run one hundred simulations where we randomly keep $N$ observations for Serial Dictatorship and the mechanism in comparison in the sample and re-calculate the treatment effect, noting the significance and magnitude of the coefficient. We run these simulations for $N \in \{25, 50, 75, 95\}$. Our simulations show that the coefficient sizes remain close, on average, to the same as in Table \ref{table:persistence-sample-analysis}. For $N = 25$, 19\%, 13\%, and 15\% of the coefficient estimates were significant for SDI, UTTC, and ASDI, respectively. This jumps to 25\%, 22\%, and 35\% at $N = 50$, and 51\%, 42\%, and 59\% at $N = 75$. Last, at $N = 95$, 55\%, 84\%, and 88\% are significant. We conclude that a minimum sample size of $N = 100$ is necessary for detecting statistically significant effects, if any exist, for our mechanisms.

\section{Data, Tables, and Figures}

We keep the following features in our Placement Files: Child Sex, Child Race American Indian or Alaska Native, Child Race Asian, Child Race Black or African American, Child Race Native Hawaiian or Other Pacific Islander, Child Race White, Child Race Unable to Determine, Child Hispanic or Latino Ethnicity, Child Has Been Clinically Diagnosed with Disability, Mental Retardation, Visually or Hearing Impaired, Physically Disabled, Emotionally Disturbed, Other Medically Diagnosed Condition Requiring Special Care, Child Has Previously Been Adopted, Age on Date of Legal Adoption, Total Number of Removals from Home, Placement Settings in Current FC Episode, Removal Manner, Removal Reason–Physical Abuse, Removal Reason–Sexual Abuse, Removal Reason–Neglect, Removal Reason–Alcohol Abuse Parent, Removal Reason–Drug Abuse Parent, Removal Reason–Alcohol Abuse Child, Removal Reason–Drug Abuse Child, Removal Reason–Child Disability, Removal Reason–Child Behavior Problem, Removal Reason–Parent Death, Removal Reason–Parent Incarceration, Removal Reason–Caretaker Inability Cope, Removal Reason–Abandonment, Removal Reason–Relinquishment, Removal Reason Inadequate Housing, Current Placement Setting, Caretaker Family Structure, 1st Principal Caretaker Year of Birth, 2nd Principal Caretaker Year of Birth, 1st Foster Caretaker Race American Indian or Alaska Native, 1st Foster Caretaker Race Asian, 1st Foster Caretaker Race Black or African American, 1st Foster Caretaker Race Native Hawaiian or Other Pacific Islander, 1st Foster Caretaker Race White, 1st Foster Caretaker Race Unable to Determine, 1st Foster Caretaker Hispanic or Latino Ethnicity, 2nd Foster Caretaker Race American Indian or Alaska Native, 2nd Foster Caretaker Race Asian, 2nd Foster Caretaker Race Black or African American, 2nd Foster Caretaker Race Native Hawaiian or Other Pacific Islander, 2nd Foster Caretaker Race White, 2nd Foster Caretaker Race Unable to Determine, 2nd Foster Caretaker Hispanic or Latino Ethnicity.

We drop the following features: The Federal Fiscal Year of this Dataset, Reporting Period End Date: Month, Reporting Period End Date: Year, State FIPS Code, State Two-Character Code, Local Agency FIPS Code, Record Number (AFCARS ID), Periodic Review Date, Foster Child’s Date of Birth, Date of First Removal, Discharge Date for Previous Removal, Date of Latest Removal from Home, Removal Transaction Date, Begin Date for Current Placement Setting, The Current Placement Setting is Outside the State, Termination Date of Parental Rights–Mom, Termination Date of Parental Rights–Dad, Date of Parents Loss of Parental Rights, Date of Discharge from Foster Care, Date that the Discharge was Recorded, Discharge Reason, Title IV-E Foster Care Payments, Title IV-E Adoption Assistance, Title IV-A TANF Payment, Only State or Other Support, Monthly Foster Care Payment, Length (Days) in Current Placement Setting, Age on the First Day of the Fiscal Year, Age at Most Recent Removal/Entry into FC, Child was in FC at the Beginning of the FFY, Child was in FC at the End of the Fiscal Year, Entered Foster Care During the Fiscal Year, Child was Discharged from Foster Care During the Fiscal Year, Child was in at Start or Entered FC During the FY, Child is Waiting for Adoption, Youth is No Longer Eligible for Foster Care Due to Age, Derived Race and Ethnicity, Derived Race, Rural Urban Continuum Code, State Foster Care Identifier.

\end{document}